\documentclass[12pt]{article}
\usepackage{graphicx}
\usepackage{enumerate}
\usepackage{natbib}
\usepackage{url} % not crucial - just used below for the URL

\newcommand{\blind}{1}

\usepackage{amssymb}
\usepackage{amsmath}
\usepackage{amsthm}
\usepackage{xfrac}
\usepackage{hyperref}
\usepackage{color}
\usepackage{algorithm}
\usepackage{algpseudocode}
\usepackage{caption}
\usepackage{subcaption}

\DeclareMathOperator*{\argmin}{arg\,min}
\DeclareMathOperator*{\argmax}{arg\,max}

\numberwithin{equation}{section}
\theoremstyle{plain}
\newtheorem{theorem}{Theorem}
\newtheorem{proposition}[theorem]{Proposition}

\begin{document}

\def\spacingset#1{\renewcommand{\baselinestretch}%
{#1}\small\normalsize} \spacingset{1}

%%%%%%%%%%%%%%%%%%%%%%%%%%%%%%%%%%%%%%%%%%%%%%%%%%%%%%%%%%%%%%%%%%%%%%%%%%%%%%

\if1\blind
{
  \title{\bf Approximating Likelihood Ratios with Calibrated Discriminative Classifiers}
  \author{Kyle Cranmer$^1$, Juan Pavez$^2$, and Gilles Louppe$^1$\hspace{.2cm} \\
          $^1$New York University\\
          $^2$Federico Santa Mar\'ia University}
  \maketitle
} \fi

\if0\blind
{
  \bigskip
  \bigskip
  \bigskip
  \begin{center}
    {\LARGE\bf Title}
\end{center}
  \medskip
} \fi

\bigskip
\begin{abstract}

% In particle physics likelihood ratio tests are established tools for statistical
% inference.  These tests are complicated by the fact that computer simulators are
% used as a generative model for the data, but they do not provide a way to
% evaluate the likelihood function. We demonstrate how discriminative classifiers
% can be used to approximate the likelihood function when a generative model for
% the data is available for training and calibration.  This offers an approach to
% parametric inference when simulators are used that is complementary to
% approximate Bayesian computation.

In many fields of science, generalized likelihood ratio tests are established tools for
statistical inference.
 At the same time, it has become increasingly common that a simulator (or
generative model) is used to describe complex processes that tie parameters
$\theta$ of an underlying theory and measurement apparatus to high-dimensional
observations $\mathbf{x}\in \mathbb{R}^p$. However, simulator often do not provide a way to evaluate
the likelihood function for a given observation $\mathbf{x}$, which motivates a new class of likelihood-free
inference algorithms.
In this paper, we show that likelihood ratios are invariant under a specific class of dimensionality reduction maps $\mathbb{R}^p \mapsto \mathbb{R}$.  As a direct consequence, we show that discriminative classifiers can be used to
approximate the generalized likelihood ratio statistic when only a generative model for the data is
available. This leads to a new machine learning-based approach to likelihood-free inference that is complementary to Approximate Bayesian Computation, and which does not require a prior on the model parameters. Experimental results
on artificial problems with known exact likelihoods illustrate the potential of the proposed method.

\end{abstract}

\noindent%
\vfill

\newpage
\spacingset{1.45} % DON'T change the spacing!

% Introduction =================================================================

% - Introduction section, mostly the same as before:
%    * likelihood free setup
%    * searches for particles at LHC as an example
%    * brief and clear description of the contribution of this paper
%      (i.e., density ratio estimation with calibrated classifiers,
%             as an alternative to other methods for parameter inference)
% - outline

\section{Introduction}
\label{sec:introduction}

% The likelihood function is the central object that summarizes the information
% from an experiment needed for inference of model parameters. The likelihood
% function is key to Bayesian inference, and many areas of science that report the
% results of classical hypothesis tests or confidence intervals use the
% (generalized or profile) likelihood ratio as a test statistic. It is
% increasingly common that a simulator (or generative model) is used to describe
% complex processes that tie parameters $\theta$ of an underlying theory and
% measurement apparatus to high-dimensional observations $x$. Directly evaluating
% the likelihood ratio in these cases is often impossible or is computationally
% impractical. Approximate Bayesian Computation (ABC) is one approach to parameter
% inference in this simulation-based or likelihood-free
% setting~\citep{Rubin1984,Tavare1997,Marin2011}. Here we consider an alternative
% approach that can also be used in a classical setting where a prior over the
% parameters is not available. In particular, we demonstrate how discriminative
% classifiers can be used to construct equivalent likelihood ratio tests when a
% generative model for the data is available for training and calibration.

The likelihood function is the central object that summarizes the information
from an experiment needed for inference of model parameters. It
is key to many areas of science that report the results of classical
hypothesis tests or confidence intervals using the (generalized or profile)
likelihood ratio as a test statistic. At the same time, with the advance of
computing technology, it has become increasingly common that a simulator (or
generative model) is used to describe complex processes that tie parameters
$\theta$ of an underlying theory and measurement apparatus to high-dimensional
observations $\mathbf{x}$. However, directly evaluating the likelihood function
in these cases is often impossible or is computationally impractical.

The main result of this paper is to show that the likelihood ratio is invariant
under dimensionality reductions $\mathbb{R}^p \mapsto \mathbb{R}$, under the assumption that the corresponding
transformation is itself monotonic with the likelihood ratio. As a direct consequence,
we derive and propose an alternative machine learning-based approach for
likelihood-free inference that can also be used in a classical (frequentist)
setting where a prior over the model parameters is not available. More
specifically, we demonstrate that discriminative classifiers can be used
to construct equivalent generalized likelihood ratio test statistics when only a generative model for
the data is available for training and calibration.

% As a concrete example, consider searches for new particles at the Large Hadron
% Collider. The simulator that is sampling from $p(x|\theta)$ is based on quantum
% field theory, a detailed simulation of the particle detector, and data
% processing algorithms that transform raw sensor data into the feature vector
% $x$~\citep{Sjostrand:2006za,Agostinelli:2002hh}.
%
% The ATLAS and CMS experiments have published  hundreds papers where the final
% result was formulated as a hypothesis test or confidence interval using a
% generalized  likelihood ratio test~\citep{Cowan:2010js}. This includes the
% discovery of the Higgs boson~\citep{Aad:2012tfa,Chatrchyan:2012ufa} and
% subsequent measurement of its properties.
%
% The bulk of the likelihood ratio tests at the LHC are based on the distribution
% of a single event-level feature that discriminates between a hypothesized
% process of interest (labeled \textit{signal}) and various other processes
% (labeled \textit{background}). Typically,  pseudo-data from the simulator are
% used to approximate the density at various parameter points, and an
% interpolation algorithm is used to approximate the parameterized
% model~\citep{Cranmer:2012sba}.
%
% To improve the statistical power of these tests, hundreds of these searches have
% utilized supervised learning to train discriminative classifiers that take
% commonly used. Recently, there has been progress in using deep
% networks~\citep{Baldi:2014kfa} and a NIPS workshop synthesizing the lessons
% learned during HiggsML~\citep{HepML}, the largest Kaggle challenge in history.

As a concrete example, let us consider searches for new particles at the Large
Hadron Collider (LHC). The simulator that is sampling from $p(\mathbf{x}|\theta)$ is
based on quantum field theory, a detailed simulation of the particle detector,
and data processing algorithms that transform raw sensor data into the feature
vector $\mathbf{x}$~\citep{Sjostrand:2006za,Agostinelli:2002hh}. The ATLAS and CMS
experiments have published  hundreds of papers where the final result was
formulated as a hypothesis test or confidence interval using a generalized
likelihood ratio test~\citep{Cowan:2010js}, including most notably the discovery
of the Higgs boson~\citep{Aad:2012tfa,Chatrchyan:2012ufa} and subsequent
measurement of its properties. The bulk of the likelihood ratio tests at the LHC
are based on the distribution of a single event-level feature that discriminates
between a hypothesized process of interest (labeled \textit{signal}) and various
other processes (labeled \textit{background}).
Typically, data generated from the
simulator are used to approximate the density at various parameter points, and
an interpolation algorithm is used to approximate the parameterized
model~\citep{Cranmer:2012sba}. In order to improve the statistical power
of these tests, hundreds of these searches have already been using supervised
learning to train classifiers to discriminate between two two discrete hypotheses based
on a high dimensional feature vector $\mathbf{x}$. The results of this paper outline
how to extend the use of discriminative classifiers for composite hypotheses (parameterized by $\theta$)
in a way
that fits naturally into the established likelihood based inference techniques.

The rest of the paper is organized as follows. In
Sec.~\ref{sec:likelihood-ratio-tests}, we first introduce the likelihood
ratio test statistic in the setting of simple hypothesis testing, and then
outline how it can be computed exactly using calibrated classifiers.
In Sec.~\ref{sec:generalized-likelihood-ratio}, we generalize the proposed
approach to the case of composite hypothesis testing and discuss directions for
approximating the statistic efficiently. We then illustrate the proposed
method in Sec.~\ref{sec:examples} and outline how it could improve
statistical analysis within the field of high energy physics. Related work
and conclusions are finally presented in Sections~\ref{sec:related} and \ref{sec:conclusions}.

% Likelihood ratio tests =======================================================

% - Introduce the mathematical problem with notations (was 1.1 and 1.2 before)
% - I would keep some of the comments of 1.3, but I feel some are not necessary

\section{Likelihood ratio tests}
\label{sec:likelihood-ratio-tests}

\subsection{Simple hypothesis testing}

Let $\mathbf{X}$ be a random vector with values $\mathbf{x} \in {\cal X}
\subseteq \mathbb{R}^p$ and let $p_\mathbf{X}(\mathbf{x}|\theta)$ denote the
density probability of $\mathbf{X}$ at value $\mathbf{x}$ under the
parameterization $\theta$. Let also assume i.i.d. observed data ${\cal D} = \{
\mathbf{x}_1, \dots, \mathbf{x}_n \}$. In the setting where one is interested in
simple hypothesis testing between a null $\theta=\theta_0$ against an alternate
$\theta=\theta_1$, the Neyman-Pearson lemma states that the likelihood ratio
\begin{equation}\label{eqn:likelihood-ratio-test}
\lambda({\cal D}; \theta_0, \theta_1) = \prod_{\mathbf{x} \in {\cal D}} \frac{ p_\mathbf{X}(\mathbf{x}|\theta_0)}{ p_\mathbf{X}(\mathbf{x}|\theta_1)}
\end{equation}
is the most powerful test statistic.

In order to evaluate $\lambda({\cal D})$, one must be able to evaluate the probability
densities $p_\mathbf{X}(\mathbf{x}| \theta_0)$ and $p_\mathbf{X}(\mathbf{x}| \theta_1)$ at any value $\mathbf{x}$. However,
it is increasingly common in science that one has a complex simulation that can
act as generative model for $p_\mathbf{X}(\mathbf{x}|\theta)$, but one cannot
evaluate the density directly. For instance, this is the case in high energy
physics~\citep{Neal:2007zz} where the simulation of particle detectors can only
be done in the forward mode.

% Approximating likelihood ratios with classifiers =============================

% - Core of the paper (was section 2 + my note on slack)
% - Derive theorem for the ratio of transformed densitites
% - Show that consistent regressors minimizing the squared error loss fulfill conditions of the theorem
% - Discuss that training on signal vs background instead of signal+background vs background is OK because this learns a function which is 1:1. => focus the capacity of the classifier; Generalization to decomposed ratios
% - Discuss about the need of calibration in case of imperfect classifiers (was first paragraphs of section 5)

\subsection{Approximating likelihood ratios with classifiers}
\label{sec:approx}

% The main result of this paper is to generalize the observation that one can form an equivalent test based on
% %\begin{equation}
% %T'(D) = \prod_{e=1}^n \frac{ p(\,s(x_e; \theta_1, \theta_0) \mid \theta_1)}{ p(\,s(x_e; \theta_1, \theta_0)\mid\theta_0)}
% %\end{equation}
% \begin{equation}\label{eq:equivLRtest}
% T'(D; \theta_0, \theta_1) = \prod_{e=1}^n \frac{ p(s_e | \theta_0)}{ p(s_e | \theta_1)}
% \end{equation}
% if
% \begin{equation}\label{eq:montonic}
% s_e = s(x_e; \theta_0, \theta_1) = m\left(\, p(x_e|\theta_0) / p(x_e|\theta_1) \,\right) \;
% %s_e = s(x_e; \theta_0, \theta_1) = m\left(\frac{ p(x_e|\theta_0)}{ p(x_e|\theta_1)} \right) \;
% \end{equation}
% where $m$ is any strictly increasing or decreasing function. This result will be proven below.
% This allows us to recast the original likelihood ratio test into an alternate form in which supervised learning is used to train the discriminative classifier $s(x; \theta_0, \theta_1)$. The discriminative classifier can be trained with data $(x,u=0)$ generated
% from $p(x|\theta_0)$ and $(x,u=1)$ generated from $p(x|\theta_1)$. In Sec.~\ref{S:GLR} we extend this result to generalized likelihood ratio tests, where it will be useful to have the classifier  parameterized in terms of $(\theta_0, \theta_1)$.
%
% Here we see that the original goal for simple hypothesis testing (i.e. to make a decision to accept or reject the null hypothesis based on the entire data set $D$) has been reformulated into a per-event classification problem. This follows from the fact that we assume the $x_e$ to be i.i.d.

The main result of this paper is to
generalize the observation that one can form a test statistic
\begin{equation}\label{eqn:likelihood-ratio-test-equiv}
\lambda'({\cal D}; \theta_0, \theta_1) = \prod_{\mathbf{x} \in {\cal D}} \frac{ p_\mathbf{U}(u=s(\mathbf{x}) | \theta_0)}{ p_\mathbf{U}(u=s(\mathbf{x}) | \theta_1)}
\end{equation}
that is strictly equivalent to \ref{eqn:likelihood-ratio-test}, provided the change
of variable $\mathbf{U} = s(\mathbf{X})$ is based
on a (parameterized) function $s$ that is strictly monotonic with the density ratio
\begin{equation}
r(\mathbf{x};\theta_0, \theta_1) = \frac{p_\mathbf{X}(\mathbf{x}|\theta_0)}{p_\mathbf{X}(\mathbf{x}|\theta_1)}.
\end{equation}
As derived below, this allows to recast the original likelihood ratio test into
an alternate form in which supervised learning can be used to build
$s(\mathbf{x})$ as a discriminative classifier.  In
Sec.~\ref{sec:generalized-likelihood-ratio} we extend this result to
generalized likelihood ratio tests, where it will be useful to have the
classifier decision function $s$ parameterized in terms of $(\theta_0, \theta_1)$.

\begin{theorem}
    \label{thm:ratio-equivalence}
    Let $\mathbf{X}$ be a random vector with values in ${\cal X} \subseteq \mathbb{R}^p$ and parameterized probability
    density $p_{\mathbf{X}}(\mathbf{x} = (x_1, ..., x_p)|\theta)$ and let
    $s : \mathbb{R}^p \mapsto \mathbb{R}$ be a function monotonic with the density ratio
    $r(\mathbf{x};\theta_0,\theta_1)$,
    for given parameters $\theta_0$ and $\theta_1$. In these conditions,
    \begin{equation}
        r(\mathbf{x};\theta_0,\theta_1) = \frac{p_\mathbf{X}(\mathbf{x}|\theta_0)}{p_\mathbf{X}(\mathbf{x}|\theta_1)} = \frac{p_\mathbf{U}(u=s(\mathbf{x})|\theta_0)}{p_\mathbf{U}(u=s(\mathbf{x})|\theta_1)},
    \end{equation}
    where $p_\mathbf{U}(u=s(\mathbf{x};\theta_0,\theta_1)|\theta)$ is the induced probability density of
    $\mathbf{U} = s(\mathbf{X};\theta_0,\theta_1)$.
\end{theorem}

\begin{proof}
Starting from the definition of the probability density function, we have
\begin{align}
% p_{\mathbf{U}}(u=s(\mathbf{x})|\theta_0) &= \frac{d}{dy} P(s(\mathbf{X}) \leq y) \nonumber \\
p_{\mathbf{U}}(u=s(\mathbf{x})|\theta_0) &= \frac{d}{du}  \int_{\{\mathbf{x}^\prime : s(\mathbf{x}^\prime) \leq u\}} p_\mathbf{X}(\mathbf{x}^\prime|\theta_0) d\mathbf{x}^\prime \nonumber \\
% &= \frac{d}{dy} \int_{\mathbb{R}^p} H(y - s(\mathbf{x}^\prime)) p_\mathbf{X}(\mathbf{x}^\prime|\theta_0) d\mathbf{x}^\prime \nonumber \\
% &= \int_{\mathbb{R}^p} \frac{d}{dy} H(y - s(\mathbf{x}^\prime)) p_\mathbf{X}(\mathbf{x}^\prime|\theta_0) d\mathbf{x}^\prime \nonumber \\
&= \int_{\mathbb{R}^p} \delta(u - s(\mathbf{x}^\prime)) p_\mathbf{X}(\mathbf{x}^\prime|\theta_0) d\mathbf{x}^\prime
\end{align}
Intuitively, this expression can be understood as the integral
over all $\mathbf{x}^\prime \in \mathbb{R}^p$ such that $s(\mathbf{x}^\prime) = u$, as picked
by the Dirac $\delta$ function. Given Theorem 6.1.5 of \citet{Hrmander1990},
it further comes
\begin{equation}
p_{\mathbf{U}}(u=s(\mathbf{x})|\theta_0) = \int_{\mathbf{x}^\prime \in \Omega_u } \frac{1}{|\nabla s(\mathbf{x}^\prime)|} p_\mathbf{X}(\mathbf{x}^\prime|\theta_0) dS_{\mathbf{x}^\prime} \label{eqn:hormander}
\end{equation}
where $\Omega_u = \{x^\prime : s(\mathbf{x}^\prime)=u\}$, $|\nabla s(\mathbf{x}^\prime)| = \sqrt{\sum_{i=1}^p |\frac{\partial}{\partial x_i} s(\mathbf{x}^\prime)|^2}$
and where $dS_{\mathbf{x}^\prime}$ is the Euclidean surface measure on $\Omega_u$.
Also, since $s(\mathbf{x})$ is monotonic with
$r(\mathbf{X}; \theta_0, \theta_1)$,
there exists an invertible function $m:\mathbb{R}^+ \mapsto \mathbb{R}$ such
that $s(\mathbf{x}) = m(r(\mathbf{X}; \theta_0, \theta_1))$.
In particular, we have
\begin{align}
\frac{p_\mathbf{X}(\mathbf{x}|\theta_0)}{p_\mathbf{X}(\mathbf{x}|\theta_1)} &= m^{-1}(s(\mathbf{x})) \nonumber \\
p_\mathbf{X}(\mathbf{x}|\theta_0) &= m^{-1}(s(\mathbf{x})) p_\mathbf{X}(\mathbf{x}|\theta_1) \label{eqn:mapping}
\end{align}
Combining equations \ref{eqn:hormander} and \ref{eqn:mapping}, the density ratio $r(\mathbf{X}; \theta_0, \theta_1)$ can be pulled out of the integral, resulting in
\begin{align}
p_{\mathbf{U}}(u=s(\mathbf{x})|\theta_0) &= \int_{\Omega_u}  \frac{1}{|\nabla s(\mathbf{x}^\prime)|} m^{-1}(s(\mathbf{x}^\prime)) p_\mathbf{X}(\mathbf{x}^\prime|\theta_1) dS_{\mathbf{x}^\prime} \nonumber \\
&= \int_{\Omega_u}  \frac{1}{|\nabla s(\mathbf{x}^\prime)|} m^{-1}(u) p_\mathbf{X}(\mathbf{x}^\prime|\theta_1) dS_{\mathbf{x}^\prime} \nonumber \\
&= m^{-1}(s(\mathbf{x})) \int_{\Omega_u}  \frac{1}{|\nabla s(\mathbf{x}^\prime)|}  p_\mathbf{X}(\mathbf{x}^\prime|\theta_1) dS_{\mathbf{x}^\prime} \nonumber \\
&= \frac{p_\mathbf{X}(\mathbf{x}|\theta_0)}{p_\mathbf{X}(\mathbf{x}|\theta_1)} \int_{\Omega_u}  \frac{1}{|\nabla s(\mathbf{x}^\prime)|}  p_\mathbf{X}(\mathbf{x}^\prime|\theta_1) dS_{\mathbf{x}^\prime}. \label{eqn:factorization}
\end{align}
Similarly, Equation~\ref{eqn:hormander} can be used to derive $p_{\mathbf{U}}(u=s(\mathbf{x})|\theta_1)$, finally yielding
\begin{align}
\frac{p_{\mathbf{U}}(u=s(\mathbf{x})|\theta_0)}{p_{\mathbf{U}}(u=s(\mathbf{x})|\theta_1)} &= \frac{p_\mathbf{X}(\mathbf{x}|\theta_0)}{p_\mathbf{X}(\mathbf{x}|\theta_1)} \frac{\int_{\Omega_u}  \frac{1}{|\nabla s(\mathbf{x}^\prime)|}  p_\mathbf{X}(\mathbf{x}^\prime|\theta_1) dS_{\mathbf{x}^\prime}}{ \int_{\Omega_u}  \frac{1}{|\nabla s(\mathbf{x}^\prime)|}  p_\mathbf{X}(\mathbf{x}^\prime|\theta_1) dS_{\mathbf{x}^\prime} } \nonumber \\
&= \frac{p_\mathbf{X}(\mathbf{x}|\theta_0)}{p_\mathbf{X}(\mathbf{x}|\theta_1)}. \label{eqn:jacob}
\end{align}
\end{proof}

% \subsection{Probabilistic classification for likelihood ratios}
% \label{sec:clf-for-ratios}

% \subsection{The fixed discriminative classification setting} For fixed
% $\theta_0$ and $\theta_1$ we can generate large samples from each model and
% train a classifier. To be concrete, let us use $p(x|\theta_0)$ to generate
% training data $(x_i,  y_i=0)$ and $p(x|\theta_1)$ to generate training data
% $(x_i , y_i=1)$. With balanced training data
% \mbox{($p(u=1)=p(u=0)=\sfrac{1}{2}$)} a quadratic loss function will lead to
% classifiers that approximate the regression function  $\hat{s}(x) \approx p(y|x) =
% p(x|\theta_1)/(p(x|\theta_0)+p(x|\theta_1))$, which is  monotonic with the
% desired per-event likelihood ratio $q(x)$. Thus, standard supervised learning
% algorithms with various surrogate loss functions lead to discriminative
% classifiers that approximate a monotonic function of per-event likelihood ratio
% $q(x)$.

In light of this result, the likelihood ratio estimation problem can now
be recast as a (probabilistic) classification problem, by noticing that the decision
function
\begin{equation}\label{eqn:best-s-clf}
s^*(\mathbf{x}) = \frac{p_{\mathbf{X}}(\mathbf{x}|\theta_1)}{p_{\mathbf{X}}(\mathbf{x} | \theta_0) + p_{\mathbf{X}}(\mathbf{x} | \theta_1)}.
\end{equation}
modeled by a classifier trained to distinguish samples $\mathbf{x} \sim p_{\theta_0}$
from samples $\mathbf{x} \sim p_{\theta_1}$ satisfies the conditions of
Thm.~\ref{thm:ratio-equivalence} (see Appendix~\ref{app:clf-for-s}).
In other words, supervised learning yields a sufficient
procedure for Thm.~\ref{thm:ratio-equivalence} to hold, guaranteeing that any
{\it universally strongly consistent} algorithm can be used for learning $s^*$.
Note however, that it is not a necessary procedure since
Thm.~\ref{thm:ratio-equivalence} holds for any monotonic function $m$ of the
density ratio,  not only for $m(r(\mathbf{x})) = (1 +
r(\mathbf{x}))^{-1}$.
Equivalently,
Thm.~\ref{thm:ratio-equivalence} shows that in the case that we learn a
classifier $s(\mathbf{x})$ which is imperfect up to a monotonic transformation
of $r(\mathbf{x})$, then one can still resort to calibration (i.e., modeling
$p_{\mathbf{U}}(u=s(\mathbf{x}))$) to compute $r(\mathbf{x})$ exactly.
For this reason, the proposed method is expected to be more robust than
directly using ${(1 - s(\mathbf{x}))}/{s(\mathbf{x})}$ as an approximate of
$r(\mathbf{x})$ (which indeed converges towards $r(x)$ when $s(x)$ tends to $s^*(x)$).

\subsection{Learning and calibrating $s$}\label{s:learning_calibration}

In order for the proposed approach to be useful in the likelihood-free setting,
we need to be able to approximate both $s(\mathbf{x})$ and
$p(s(\mathbf{x})|\theta)$ based on a finite number of samples $\{\mathbf{x}_i\}$
drawn from the generative model $p(\mathbf{x}|\theta)$.

As outlined above, any consistent probabilistic classification algorithm can be
used for learning an approximate  map $\hat{s}(\mathbf{x})$ of
Eqn.~\ref{eqn:best-s-clf}. In the common case where the density ratio is
expected to smoothly vary around $\mathbf{x}$, we would however recommend
learning models whose output value $\hat{s}(\mathbf{x})$ also smoothly varies
around $\mathbf{x}$, such as neural networks. For small training sets, tree-based
methods are not expected to work so well for this use case, since they usually
model $\hat{s}(\mathbf{x})$ as a non-strictly monotonic composition of step functions.
In such cases where $\hat s(\mathbf{x})$ is not monotonic with $r(\mathbf{x})$,
the induced probability does not factorize as in Eqn.~\ref{eqn:factorization},
leading to artifacts in the resulting approximation of the density ratio.
Provided enough training data, accurate results can however still be achieved,
given the universal approximator capacity of tree-based models.

Given a reduction map $s$, our results show that a statistic equivalent to the
likelihood ratio can be constructed, provided $p(s(\mathbf{x})|\theta)$ can be
evaluated. Again, we do not have a direct and exact way for evaluating this
density, but an approximation $\hat p(\hat s(\mathbf{x})| \theta)$ can be built
instead, e.g. using density estimation algorithms, such as
histograms or kernel density estimation applied to $\{\hat
s(\mathbf{x}_i)\}$, where the $\{\mathbf{x}_i\}$  are drawn from the generative
model. Most notably, learning such an approximation of $p(s(\mathbf{x})|\theta)$
is a much simpler problem than learning $p(\mathbf{x}|\theta)$, since the
reduction $s$ projects $\mathbf{x}$ into a one-dimensional space in which only
the (simpler) informative content of $r(\mathbf{x})$ is preserved.

An alternative approach for calibration is to approximate the density ratio
$r(\hat{s}(\mathbf{x}))$ directly.
For example, isotonic regression, which is commonly used to transform the classifier score  $\hat{s}(\mathbf{x})$ into $\hat{s}_\textrm{iso}(\mathbf{x})$ that more accurately reflect  the posterior probability $s^*(\mathbf{x})$
of Eqn.~\ref{eqn:best-s-clf}, can be used for calibration. This is done by inverting the relationship ${r}(\mathbf{x}) = (1-s^*(\mathbf{x}))/s^*(\mathbf{x})$ to obtain $\hat{r}(\hat{s}(\mathbf{x})) = (1-\hat{s}_\textrm{iso}(\mathbf{x}))/\hat{s}_\textrm{iso}(\mathbf{x})$.
Additionally, Sec.~\ref{sec:related} describes related work in which the ratio $\hat{r}(\mathbf{x})$ is estimated directly on the feature space $\mathcal{X}$.

%\textbf{Gilles, should we add more about direct density ratio estimation (currently in 5. Related Work) here?}

%An alternative approach for calibration is to approximate the density ratio
%$r(\hat{s}(\mathbf{x}); \theta_0, \theta_1)$ directly.
%%There are various strategies.
%For instance, isotonic regression, which is commonly used to transform the classifier score  $\hat{s}(\mathbf{x})$ into $\hat{s}_\textrm{iso}(\mathbf{x})$ that more accurately reflect  the posterior probability $s^*(\mathbf{x})$
%of Eqn.~\ref{eqn:best-s-clf}, can be used for calibration. This is done by inverting the relationship ${r}(\mathbf{x}; \theta_0, \theta_1) = (1-s^*(\mathbf{x}))/s^*(\mathbf{x})$ to obtain $\hat{r}(\mathbf{x}; \theta_0, \theta_1) = (1-\hat{s}_\textrm{iso}(\mathbf{x}))/\hat{s}_\textrm{iso}(\mathbf{x})$.
%\textbf{Gilles, should we add more about direct density ratio estimation (currently in 5. Related Work) here?}

%. , which for balanced training samples has the form
%$\hat{p}(\theta_1 | \hat{s}(\mathbf{x})) \approx p(\hat{s}(\mathbf{x}) | \theta_0) / (p(\hat{s}(\mathbf{x}) | \theta_1) + p(\hat{s}(\mathbf{x}) | \theta_1)) = 1/(1+r(\hat{s}(\mathbf{x}); \theta_0, \theta_1)))$.

One strength of the proposed approach is that it factorizes the approximation of
the dimensionality reduction ($\hat{s}(\mathbf{x}) \approx s(\mathbf{x})$) from the
calibration procedure ($\hat p(\hat s(\mathbf{x})| \theta) \approx
p(\hat{s}(\mathbf{x})|\theta)$ or $\hat{r}(\hat{s}(\mathbf{x})) \approx r(\hat{s}(\mathbf{x}))$). Thus, even if the classifier does a poor job at
learning the optimal decision function~\ref{eqn:best-s-clf} and, therefore, at
reproducing the level sets of the per-sample likelihood ratio, the density of
$\hat{s}$ can still be well calibrated. In that case, one might loose power, but
the resulting inference will still be valid. This point was made by
\cite{Neal:2007zz} and is well appreciated by the particle physics community
that typically takes a conservative attitude towards the use of machine learning
classifiers precisely due to concerns about the calibration of $p$-values in the
face of nuisance parameters associated to the simulator.

\section{Generalized likelihood ratio tests}
\label{sec:generalized-likelihood-ratio}

Thus far we have shown that the target likelihood ratio
$r(\mathbf{x};\theta_0,\theta_1)$
with high dimensional features $\mathbf{x}$ can be reproduced via the univariate
densities $p(s(\mathbf{x})|\theta_0)$ and $p(s(\mathbf{x})|\theta_1)$ if the
reduction $s(\mathbf{x})$ is monotonic with $r(\mathbf{x};\theta_0,\theta_1)$.
We now generalize from the ratio of two simple hypotheses specified by
$\theta_0$ and $\theta_1$ to the case of composite hypothesis testing where
$\theta$ are continuous model parameters.

\subsection{Composite hypothesis testing}
\label{sec:composite-hypothesis-testing}

In the case of composite hypotheses $\theta \in \Theta_0$ against an alternative
$\theta \in \Theta_1$ (such that $\Theta_0 \cap \Theta_1 = \emptyset$ and $\Theta_0 \cup \Theta_1 = \Theta$), the
generalized likelihood ratio test, also known as the profile likelihood ratio
test, is commonly used
\begin{equation}\label{eqn:generalized-lr}
\Lambda(\Theta_0) =  \frac{ \sup_{\theta \in \Theta_0} p({\cal D} | \theta)}{ \sup_{\theta \in \Theta} p({\cal D} | \theta)} \; .
\end{equation}
%\begin{equation}\label{eqn:generalized-lr}
%\Lambda({\cal D}; \Theta_0, \Theta) =  \frac{ \sup_{\theta \in \Theta_0} p({\cal D} | \theta)}{ \sup_{\theta \in \Theta} p({\cal D} | \theta)} \; .
%\end{equation}
This generalized likelihood ratio can be used both for hypothesis tests in the
presence of nuisance parameters or to create confidence intervals with or
without nuisance parameters.  Often, the parameter vector is broken into two
components $\theta=(\mu,\nu)$, where the $\mu$ components are considered
parameters of interest while the $\nu$ components are considered nuisance
parameters. In that case $\Theta_0$ corresponds to all values of $\nu$ with
$\mu$ fixed.

Evaluating the generalized likelihood ratio as defined by
Eqn.~\ref{eqn:generalized-lr} requires finding for both the numerator and the
denominator the maximum likelihood estimator
\begin{equation}\label{eq:mle}
    \hat{\theta} = \argmax_\theta p({\cal D} | \theta).
\end{equation}
Again, this is made difficult in the likelihood-free setting and it is not
obvious that we can find the same estimators if we are working instead with
$p(s(\mathbf{x})|\theta)$. Fortunately, there is a construction
based on $s$ that works: the maximum likelihood estimate of Eqn.~\ref{eq:mle} is
the same as the value that maximizes the likelihood ratio with respect to
$p({\cal D}|\theta_1)$, for some fixed value of $\theta_1$ chosen such that the support of $p(\mathbf{x}|\theta_1)$ covers the support of $p(\mathbf{x}|\theta)$.
This allows us to
use Thm.~\ref{thm:ratio-equivalence} to reformulate the maximum likelihood
estimate (MLE) as
\begin{align}\label{eq:mle_withs}
\hat{\theta} &= \argmax_\theta  p({\cal D} | \theta) \nonumber \\
&= \argmax_\theta  \prod_{\mathbf{x} \in {\cal D}} \frac{p(\mathbf{x}| \theta)}{p(\mathbf{x}|\theta_1)} \nonumber \\
&= \argmax_\theta  \prod_{\mathbf{x} \in {\cal D}} \frac{p(s(\mathbf{x}; \theta, \theta_1) | \theta)}{p(s(\mathbf{x}; \theta, \theta_1) |\theta_1)} \;,
\end{align}
where $s(\mathbf{x};\theta,\theta_1)$ denotes a \textit{parameterized}
transformation $s$ of $\mathbf{X}$ in terms of $(\theta,\theta_1)$ that is monotonic
with $r(\mathbf{x};\theta,\theta_1)$. Note that it is important that we include
the denominator $p(s(\mathbf{x}; \theta, \theta_1) |\theta_1)$ because this
cancels Jacobian factors that vary with $\theta$.

Finally, once the maximum likelihood estimates have been found for both the numerator
and denominator of Eqn.~\ref{eqn:generalized-lr}, the generalized likelihood
ratio can be estimated as outlined in Sec.~\ref{sec:approx}
for simple hypothesis testing.

% For completeness, the proposed method is summarized in
% Alg.~\ref{alg:training}, while further details for training and calibrating
% efficiently a family $s(\mathbf{x};\theta_0,\theta_1)$ of classifiers are
% discussed in sections \ref{sec:param-clf} and \ref{sec:param-calibration}.

\subsection{Parameterized classification}
\label{sec:param-clf}

In order to provide parameter inference in the likelihood-free setting as
described above, we must train a family $s(\mathbf{x};\theta_0,\theta_1)$ of
classifiers parameterized by $\theta_0$ and $\theta_1$, the parameters
associated to the null and alternate hypotheses, respectively. While this could
be done independently for all $\theta_0$ and $\theta_1$, using the procedure
outlined in Sec.~\ref{sec:likelihood-ratio-tests}, it is desirable and
convenient to have a smooth evolution of the classification score as a function
of the parameters. For this reason, we anticipate a single learning stage based
on training data with input $(\mathbf{x}, \theta_0, \theta_1)_i$ and target
$y_i$, as outlined in Alg.~\ref{alg:training}. Somewhat unusually, the
unknown values of the parameters are taken as input to the classifier; their
values will be specified via the enveloping (generalized) likelihood ratio of
Eqn.~\ref{eqn:generalized-lr}.  In this way, the parameterized classifier
now models the distribution of the output $y$ conditional to $(\mathbf{x}, \theta_0, \theta_1)$,
for any $\mathbf{x}$ and any combination of parameter values $\theta_0, \theta_1$.

While the optimal decision function \ref{eqn:best-s-clf} is expected to be
learned for the parameter values $\theta_0$ and $\theta_1$ selected in
Alg.~\ref{alg:training}, it is not clear whether the optimal decision
function can be expected for data generated from  $\theta'_0$ and $\theta'_1$
never jointly encountered during learning. Similarly, it is not clear how the
limited capacity of the classifier may impact the performance of the resulting
parameterized decision function.
Preliminary exploration by \cite{Baldi:2016fzo} shows that
 a uniform grid scan over parameter space is an effective practical approach;
 however, we introduce the distributions $\pi_{\Theta_0}$ and $\pi_{\Theta_1}$
into the Alg.~\ref{alg:training} to allow for a more
 sophisticated sampling strategy.
% this formulation of Alg.~\ref{alg:training} allows for a more
% sophisticated sampling strategy.
% These issues are left as an area for future work.
%For this reason, we introduce the distributions $\pi_{\Theta_0}$ and $\pi_{\Theta_1}$
%into the Alg.~\ref{alg:training} to allow for a more
% sophisticated sampling strategy.

\begin{algorithm}[t]
\caption{Learning a parameterized classifier.}\label{alg:training}
\begin{algorithmic}
    \State ${\cal T := \{ \}}$;
    \While{ $\text{size}({\cal T}) < N $ }
        %\State Select or draw $\theta_0$, $\theta_1$ from $\Theta_0$, $\Theta_1$;
        \State Draw $\theta_0 \sim \pi_{\Theta_0}$;
	    \State Draw $\mathbf{x} \sim p(\mathbf{x}|\theta_0)$;
		\State ${\cal T} := {\cal T} \cup \{ ((\mathbf{x}, \theta_0, \theta_1), y=0) \}$;
        \State Draw $\theta_1 \sim \pi_{\Theta_1}$;
		\State Draw $\mathbf{x} \sim p(\mathbf{x}|\theta_1)$;
		\State ${\cal T} := {\cal T} \cup \{ ((\mathbf{x}, \theta_0, \theta_1), y=1) \}$;
    \EndWhile
    \State Learn a single classifier $s(\mathbf{x}; \theta_0, \theta_1)$ from ${\cal T}$.
\end{algorithmic}
\end{algorithm}

\subsection{Parameterized calibration}
\label{sec:param-calibration}

Once the parameterized classifier $\hat{s}(\mathbf{x}; \theta_0, \theta_1)$ is trained,
we can use the generative model together with one of the calibration strategies
discussed in Sec.~\ref{s:learning_calibration} for particular values of $\theta_0$ and $\theta_1$.
%For instance,
%together with a
%univariate density estimation technique (e.g. histograms or kernel density
%estimation) to approximate $p(s|\theta)$ for specific parameter
%points.
For a single parameter point $\theta$, this is a tractable univariate density
estimation problem. The challenge comes from the need to calibrate this density
for all values of $\theta$. A straightforward approach would be to run the
generative model on demand for any particular value of $\theta$. In the context
of a likelihood fit this would mean that the optimization algorithm that is
trying to maximize the likelihood with respect to $\theta$ needs access to the
generative model $p(\mathbf{x}|\theta)$.  This is the strategy used for the examples presented in Sec.~\ref{sec:examples}.

Calibrating the density on-demand can be  impractical when the generative
model is computationally expensive or has high-latency (for instance some human
intervention is required to reconfigure the generative model). In high energy physics, where
it is common to calibrate the distribution of a fixed classifier. There the strategy is to interpolate
the distribution between discrete values of $\theta$ in order to produce a continuous parameterization for
$p(s | \theta)$~\citep{read1999linear,Cranmer:2012sba,baak2015interpolation}.
One can easily imagine a number of approaches to parameterized calibration and the
relative merits of those approaches will depend critically on the dimensionality of $\theta$ and the
computational cost of the generative model. We leave a more general strategy for
this overarching optimization problem as an area of future work.

% There is a mild technical challenge in embedding the classifier into the
% likelihood. In the case of a fixed classifier $\hat s(x)$ it is possible to
% pre-compute $\hat s_e=\hat s(x_e)$ and never refer back to the original features
% $x_e$. In the parameterized setting it is not possible to pre-compute $\hat
% s(x_e; \theta_0, \theta_1)$ since one does not know the true values of
% $\theta_0$ and $\theta_1$. Thus one must implement the embedding of the
% classifier as in Eqn.~\ref{eq:embedding}.
%  A concrete realization of this has
% been performed for probability models implemented with the \texttt{RooFit}
% probabilistic programing language and  classifiers implemented with
% \texttt{scikit-learn} and
% \texttt{TMVA}~\citep{Verkerke:2003ir,scikit-learn,Hocker:2007ht}.

% Applications =================================================================

% - typical usage in HEP with a concrete experimental example.
%   * illustrate the method for several classifiers (new)
%    * illustrate the method for several calibration routines (new)
%    * compare with other methods (new)
% - ratio of mixtures and decomposition (special case of generalized likelihood rario tests when parameters only affect the mixture coefficients => no need for parameterized classifiers)
%   discuss about focusing the capacity of the classifier
% - measure of particles properties, with a concrete example. Comparison with other likelihood-free inference methods (new)

\subsection{Mixture models}
\label{sec:mixture}

In the special case of (simple or composite) hypothesis testing between
models defined as known mixtures of several components, i.e. when $p(\mathbf{x}|\theta)$ can be written as
\begin{equation}
p(\mathbf{x}|\theta)=\sum_c w_c(\theta) p_c(\mathbf{x}| \theta),
\end{equation}
the target likelihood ratio can be formulated in terms of pairwise
classification problems. Specifically, we can write
\begin{align}
\frac{p(\mathbf{x}|\theta_0)}{p(\mathbf{x}|\theta_1)} &= \frac{\sum_c w_c(\theta_0) p_c(\mathbf{x}| \theta_0)}{\sum_{c'} w_{c'}(\theta_1) p_{c'}(\mathbf{x}| \theta_1)} \nonumber \\
&= \sum_c \left[ \sum_{c'} \frac{ w_{c'}(\theta_1)}{w_c(\theta_0)} \frac{ p_{c'}(\mathbf{x}| \theta_1)}{  p_c(\mathbf{x}| \theta_0)}  \right]^{-1} \nonumber \\
&= \sum_c \left[ \sum_{c'} \frac{ w_{c'}(\theta_1)}{w_c(\theta_0)} \frac{ p_{c'}(s_{c,c'}(\mathbf{x};\theta_0, \theta_1)| \theta_1)}{ p_c(s_{c,c'}(\mathbf{x};\theta_0, \theta_1)| \theta_0)}  \right]^{-1}. \label{eq:decomposedResult}
\end{align}
The second line is a trivial, but a useful decomposition into pairwise
density ratio sub-problems between $p_{c'}(\mathbf{x}|\theta_1)$ and
$p_c(\mathbf{x}|\theta_0)$.  The third line uses
Thm.~\ref{thm:ratio-equivalence} to relate the high-dimensional likelihood
ratio into an equivalent calibrated likelihood ratio based on the univariate
density of the corresponding classifier.

In applications where mixture models are commonly used, this decomposition allows
one to construct better likelihood ratio estimates since it allows the classifiers
$s_{c,c'}$ to focus on simpler sub-problems, for which higher accuracy is
expected.

Finally, as a technical point, in the situation where the only free parameters
of the  model are the mixture coefficients $w_c$, the distributions
$p_{c}(s_{c,c'}(\mathbf{x};\theta_0, \theta_1)| \theta)$ are independent of
$\theta$. For this reason, sub-ratios $r_{c, c'}(\mathbf{x}; \theta_0, \theta_1) = \frac{
p_{c'}(s_{c,c'}(\mathbf{x};\theta_0, \theta_1)|\theta_1)}{
p_c(s_{c,c'}(\mathbf{x};\theta_0, \theta_1)|\theta_0)}$ simplify to $\frac{
p_{c'}(s_{c,c'}(\mathbf{x}))}{ p_c(s_{c,c'}(\mathbf{x}))}$, which can be
pre-computed without the need of parameterized classification or
calibration.

\subsection{Diagnostics}\label{S:diagnostics}

While Thm.~\ref{thm:ratio-equivalence} states that the likelihood ratio $r(\mathbf{x}; \theta_0, \theta_1)$ is
invariant under the dimensionality reduction $s(\mathbf{x}; \theta_0, \theta_1)$ provided that it is monotonic with
$r(\mathbf{x}; \theta_0, \theta_1)$ itself and we know that any universally strongly consistent algorithm can be
used to learn such a function, we know that in practice $\hat{r}(\hat{s}(\mathbf{x}; \theta_0, \theta_1))$ will not be exact.
Thus, it is crucial that to have some diagnostic procedures to assess the quality of this approximation.
This is complicated by the fact that in the likelihood-free setting, we don't have access to the true likelihood ratio.
Below we consider two such diagnostic procedures that can be implemented in the likelihood-free setting. We illustrate these diagnostic procedures in Fig.~\ref{fig:diagnostics}.

The first diagnostic procedure is related to the procedure for finding the MLE $\hat{\theta}$ in Eqn.~\ref{eq:mle_withs}. As pointed out there it is important that one maximizes the likelihood ratio as the surface integral and Jacobian factors related to the dimensionality reduction only cancel in the ratio (see Eqn.~\ref{eqn:factorization}).
Importantly, they also only cancel if the reduction map satisfies the assumptions of Thm.~\ref{thm:ratio-equivalence}. Moreover, the resulting value of $\hat{\theta}$ should be independent of the value of $\theta_1$ used in the denominator of the likelihood ratio. Similarly, we have
%$\log\Lambda(\theta) = \log p(\mathcal{D}|\theta)/p(\mathcal{D}|\hat\theta) = \log p(\mathcal{D}|\theta)/p(\mathcal{D}|\theta_1) -\log p(\mathcal{D}|\hat\theta)/p(\mathcal{D}|\theta_1)$ for all values of $\theta_1$.
\begin{equation}\label{eq:theta_1_independence}
\log\Lambda(\theta) = \log \frac{p(\mathcal{D}|\theta)}{p(\mathcal{D}|\hat\theta)} = \log \frac{p(\mathcal{D}|\theta)}{p(\mathcal{D}|\theta_1)} -\log \frac{p(\mathcal{D}|\hat\theta)}{p(\mathcal{D}|\theta_1)}
\end{equation}
for all values of $\theta_1$. Thus, by explicitly checking the independence of these quantities on $\theta_1$ we indirectly probe the quality of the approximation $\hat{r}(\hat{s}(\mathbf{x}; \theta_0, \theta_1)) \approx {r}({s}(\mathbf{x}; \theta_0, \theta_1))$.

The second diagnostic procedure leverages the connection of this technique to direct density ratio estimation and its application to covariate shift and importance sampling. The idea is simple: we test the relationship $p(\mathbf{x}|\theta_0) = p(\mathbf{x}|\theta_1) {r}({s}(\mathbf{x}; \theta_0, \theta_1))$ with the approximate ratio $\hat{r}(\hat{s}(\mathbf{x}; \theta_0, \theta_1))$ and samples drawn from the generative model. More specifically, we can train a classifier to distinguish between unweighted samples from $p(\mathbf{x}|\theta_0)$ and samples from $p(\mathbf{x}|\theta_1)$ weighted by $\hat{r}(\hat{s}(\mathbf{x}; \theta_0, \theta_1))$.
If the classifier can distinguish between the distributions, then $\hat{r}(\hat{s}(\mathbf{x}; \theta_0, \theta_1))$ is not a good approximation of ${r}({s}(\mathbf{x}; \theta_0, \theta_1))$. In contrast, if the classifier is unable to distinguish between the two distributions, then either $\hat{r}(\hat{s}(\mathbf{x}; \theta_0, \theta_1))$ is a good approximation or the discriminator is not effective. The two situations can be disentangled to some degree by training another classifier to distinguish between an unweighted distribution of samples from $p(\mathbf{x}|\theta_1)$.

\section{Examples and applications}
\label{sec:examples}

In this section, we illustrate the proposed method on two representative
examples where the exact likelihood is known and then discuss its application to high energy physics. The code used to produce the results and extended details for these examples is available in Ref.~\citep{carl}, which utilizes the classification and calibration routines in scikit-learn~\citep{scikit-learn}. %\footnote{\url{https://github.com/diana-hep/carl}}.

\subsection{Likelihood ratios of mixtures of normals}

As a simple and illustrative example, let us first consider the approximation of
the log-likelihood ratio $\log \left( r(\mathbf{x};\gamma=0.05,\gamma=0) \right)$ between the 1D mixtures
$p(\mathbf{x}|\gamma=0.05)$ and $p(\mathbf{x}|\gamma=0)$ defined as
\begin{align}
p(\mathbf{x}|\gamma) &= (1-\gamma)\frac{p_{c_0}(\mathbf{x}) +  p_{c_1}(\mathbf{x})}{2}   + \gamma \, p_{c_2}(\mathbf{x}),
\end{align}
where $p_{c_0} := {\cal N}(\mu=-2, \sigma^2=0.25^2)$, $p_{c_1} := {\cal N}(\mu=0, \sigma^2=4)$,
$p_{c_2} := {\cal N}(\mu=1, \sigma^2=0.25)$. Samples drawn for the nominal value $\gamma=0.05$ are shown in
Figure~\ref{fig:1a} and used later for inference.

Figure~\ref{fig:p_sij} shows the intermediate stages for the decomposition described in Sec.~\ref{sec:mixture}.
The blue and green curves show $p_{c'}(\hat{s}_{c,c'}(\mathbf{x}))$,  the distributions for the score for the sub-classifiers for the three pair-wise comparisons of the mixture components. The red curves in Fig.~\ref{fig:p_sij} show the approximation of the density ratio (rescaled as $(1 +
\hat{r}(\mathbf{x}))^{-1}$) obtained from those distributions.

%Using the method outlined in Sec.~\ref{sec:approx},

Figures~\ref{fig:1c} and \ref{fig:1d} show the approximate $\log  \hat{r}(\mathbf{x}) $ as a function of $\mathbf{x}$ using a
2-layer neural network and a random forest for the classifier $\hat{s}(\mathbf{x})$.
The neural network provides a smoother approximation, while the random has some artifacts due to the fact that the
decision function is piece-wise constant.
%ratio when using respectively a linear model, a 2-layer neural network and a random
%forest.
The blue curves show the exact $\log {r}(\mathbf{x})$, the green curves show
$\log \left( (1-\hat s(\mathbf{x}))/\hat s(\mathbf{x}) \right)$
without calibration, while the red curve is the improved approximation $\log \hat{r}(\mathbf{x})$ calibrated
using histograms. Finally, the cyan curve shows the approximated log-likelihood ratio when decomposing the mixture, as
seen in Fig.~\ref{fig:p_sij}.  By leveraging the fact that densities are mixtures, the capacity
of the underlying classifiers can be more effectively focused on easier classification tasks,
resulting as expected in even more accurate approximations.

\begin{figure}
    \centering
    \begin{subfigure}[b]{0.4\textwidth}
        \includegraphics[clip, trim=0.5cm 0.5cm 0.5cm 0.5cm, width=\textwidth]{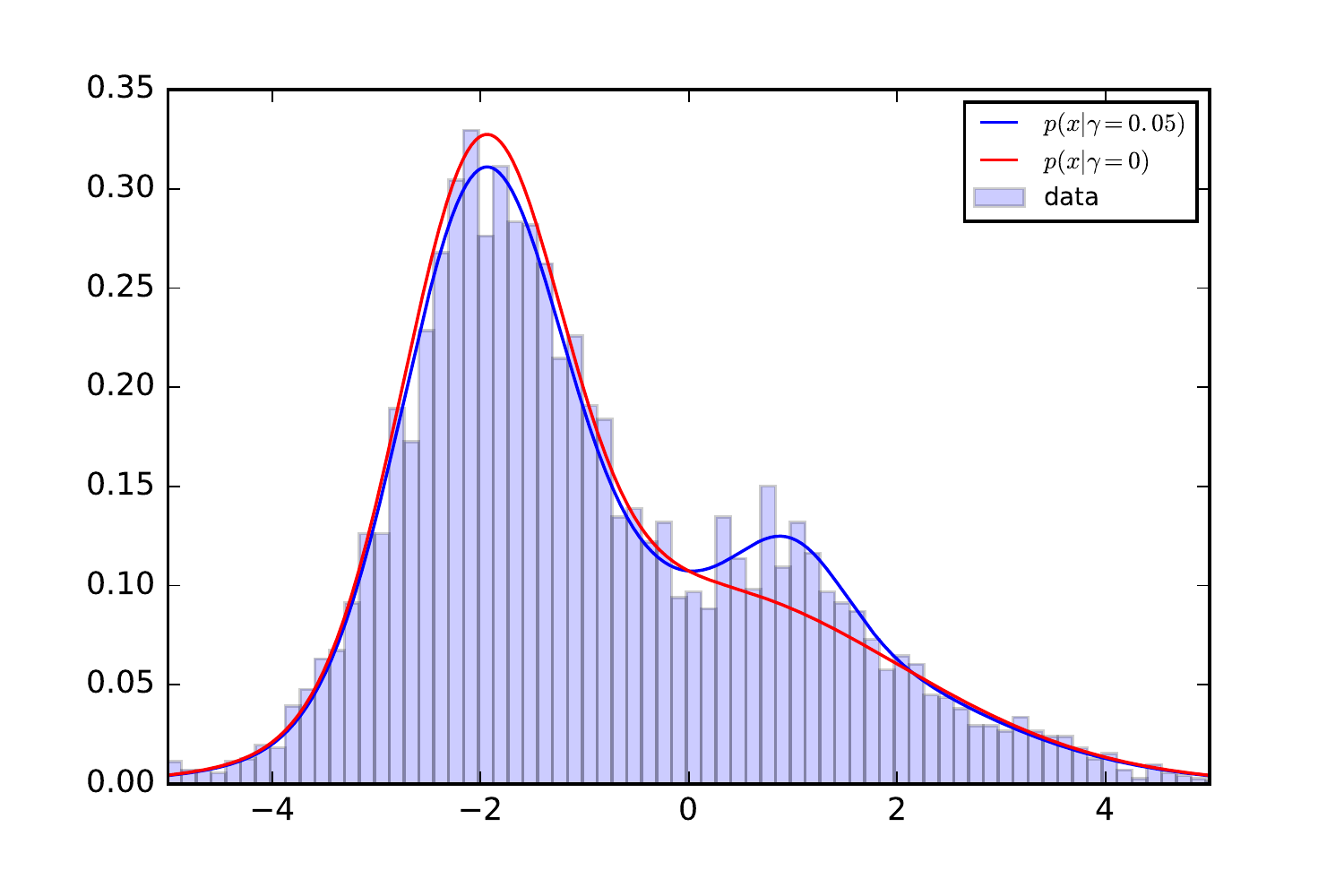}
        \caption{$p(\mathbf{x}|\gamma)$  for $\gamma=0.05$ and $\gamma=0$}
        \label{fig:1a}
    \end{subfigure}
    ~
    \begin{subfigure}[b]{0.4\textwidth}
        \includegraphics[clip, trim=0.5cm 0.5cm 0.5cm 0.5cm,width=\textwidth]{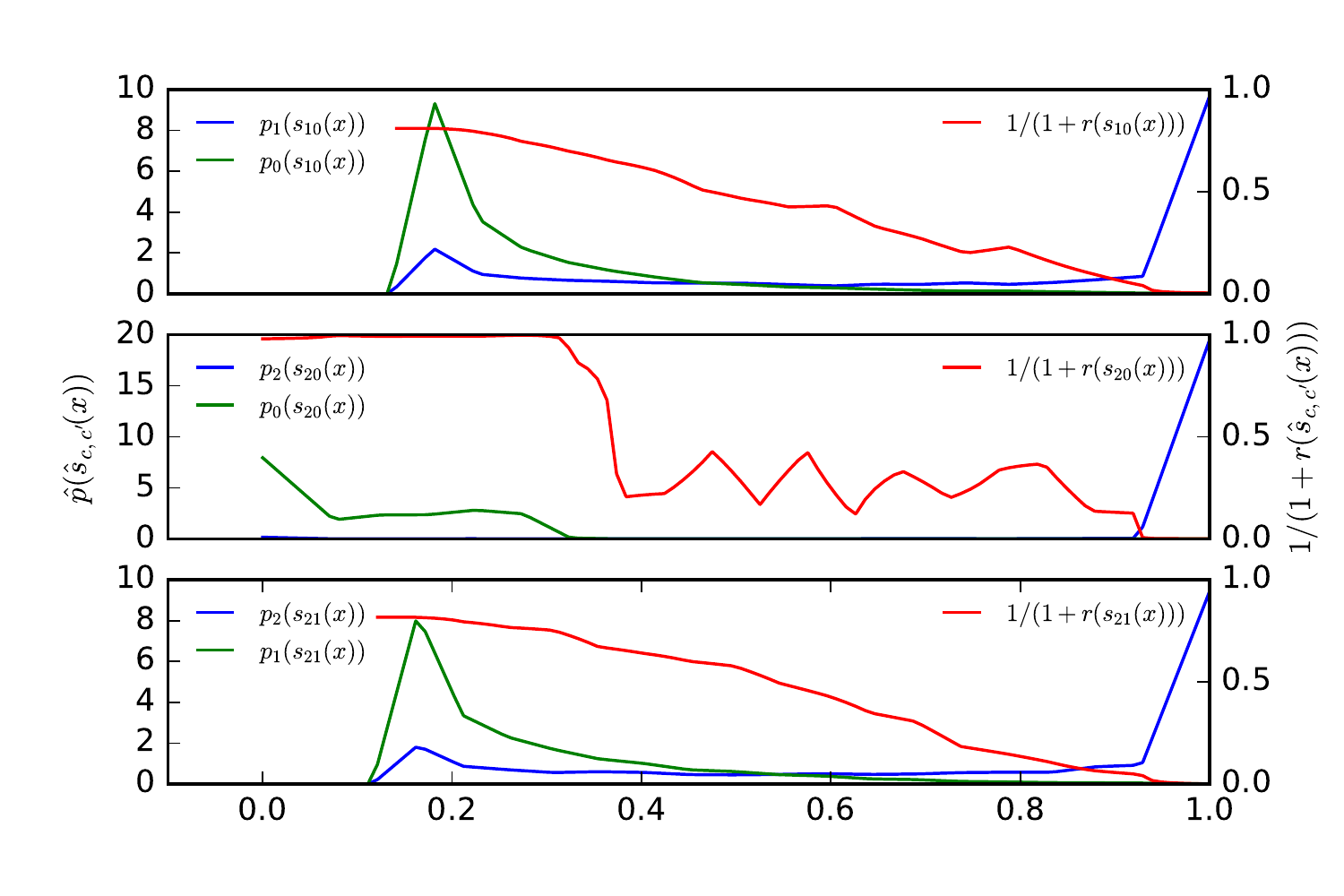}
        \caption{$\hat{p}_c(\hat{s}_{c,c^\prime}(\mathbf{x}) )$ and $(1+\hat{r}(\hat{s}_{c,c^\prime}(\mathbf{x}))^{-1}$}
        \label{fig:p_sij}
    \end{subfigure}

    \vspace{1em}

    \begin{subfigure}[b]{0.4\textwidth}
        \includegraphics[clip, trim=0.5cm 0.5cm 0.5cm 0.5cm,width=\textwidth]{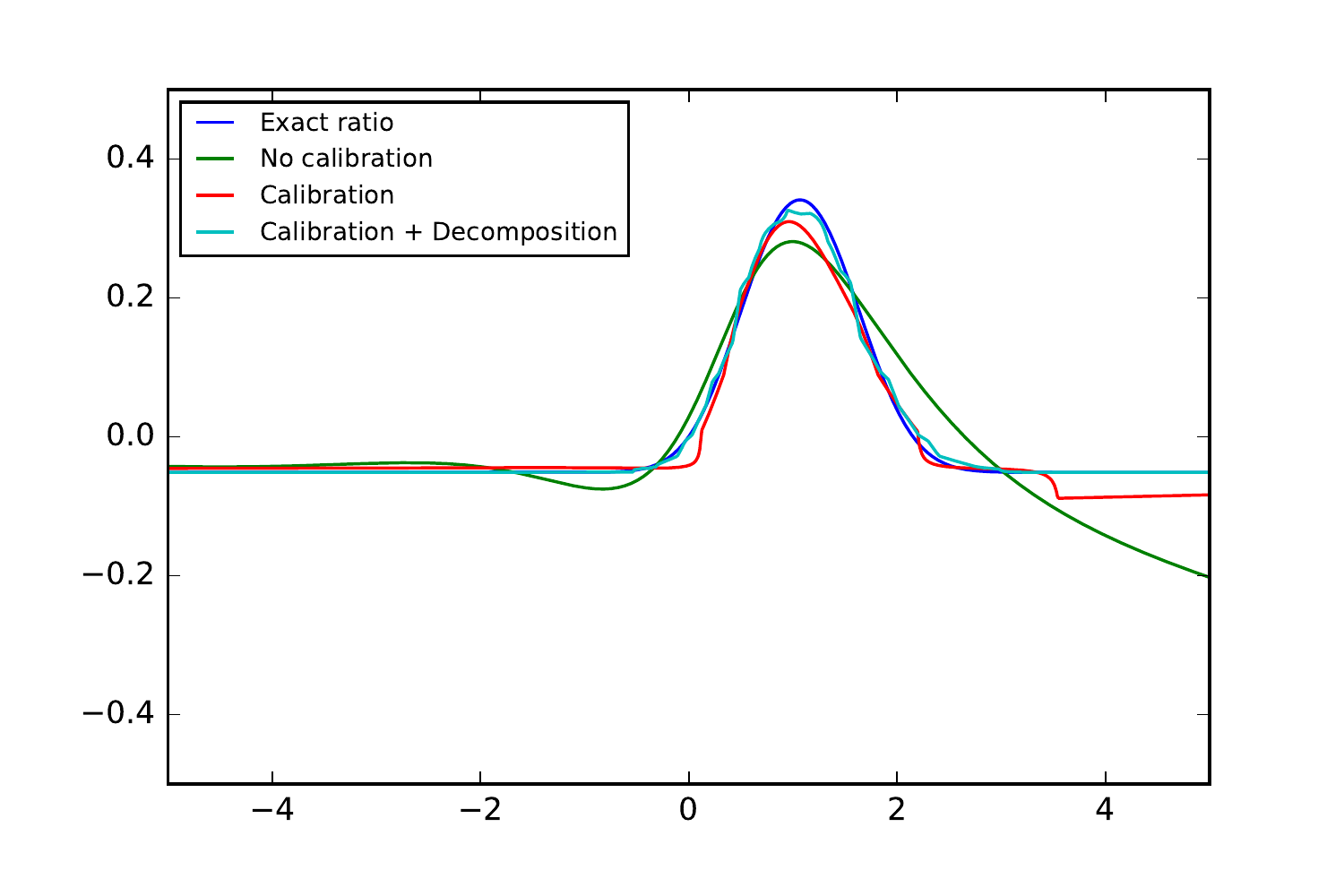}
        \caption{$\log \hat r(\hat{s}(\mathbf{x}))$ using neural network}
        \label{fig:1c}
    \end{subfigure}
    ~
    \begin{subfigure}[b]{0.4\textwidth}
        \includegraphics[clip, trim=0.5cm 0.5cm 0.5cm 0.5cm,width=\textwidth]{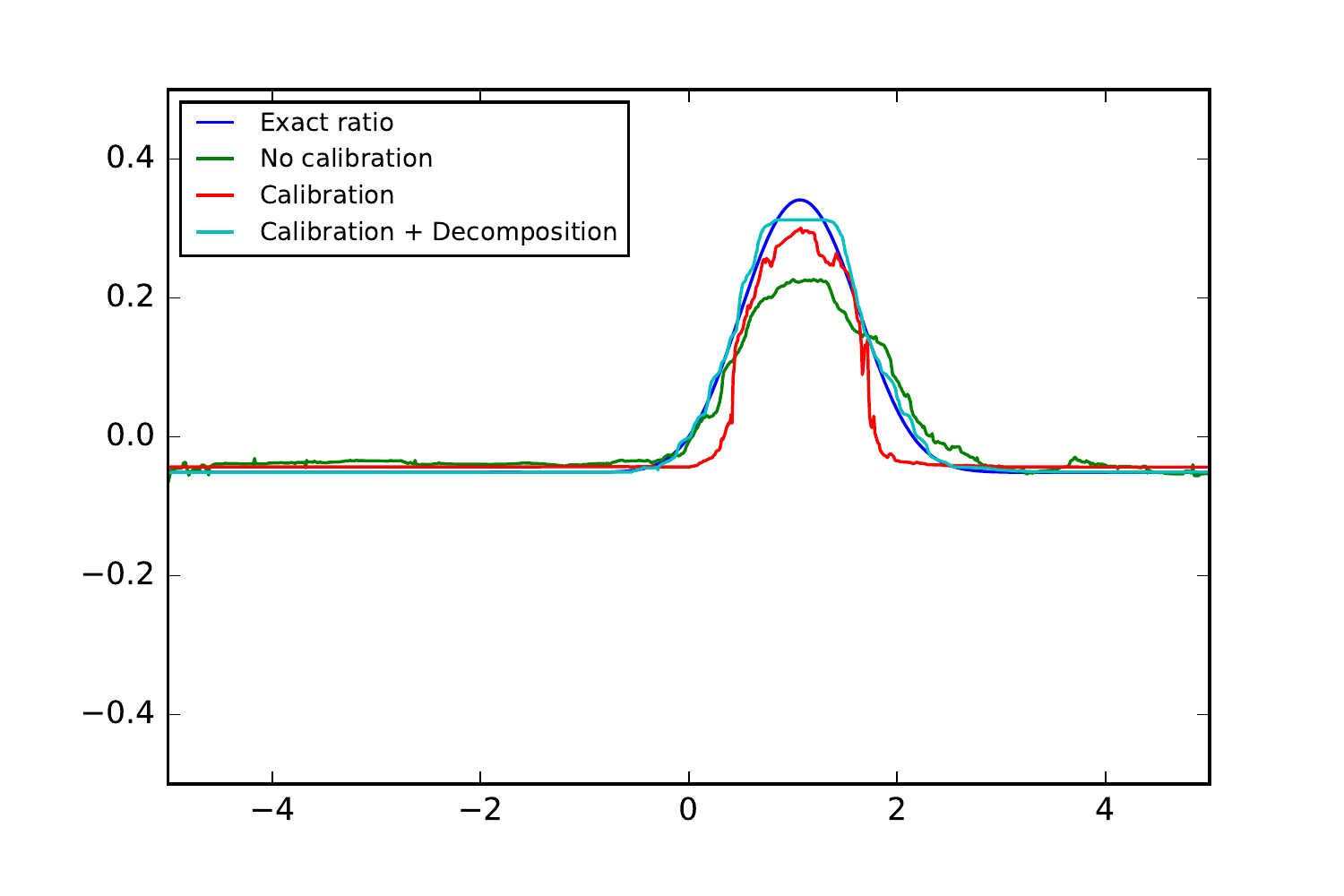}
        \caption{$\log \hat r(\hat{s}(\mathbf{x}))$ using random forest}
        \label{fig:1d}
    \end{subfigure}

    \vspace{1em}

    \begin{subfigure}[b]{0.4\textwidth}
        \includegraphics[clip, trim=0.5cm 0.5cm 0.5cm 0.5cm,width=\textwidth]{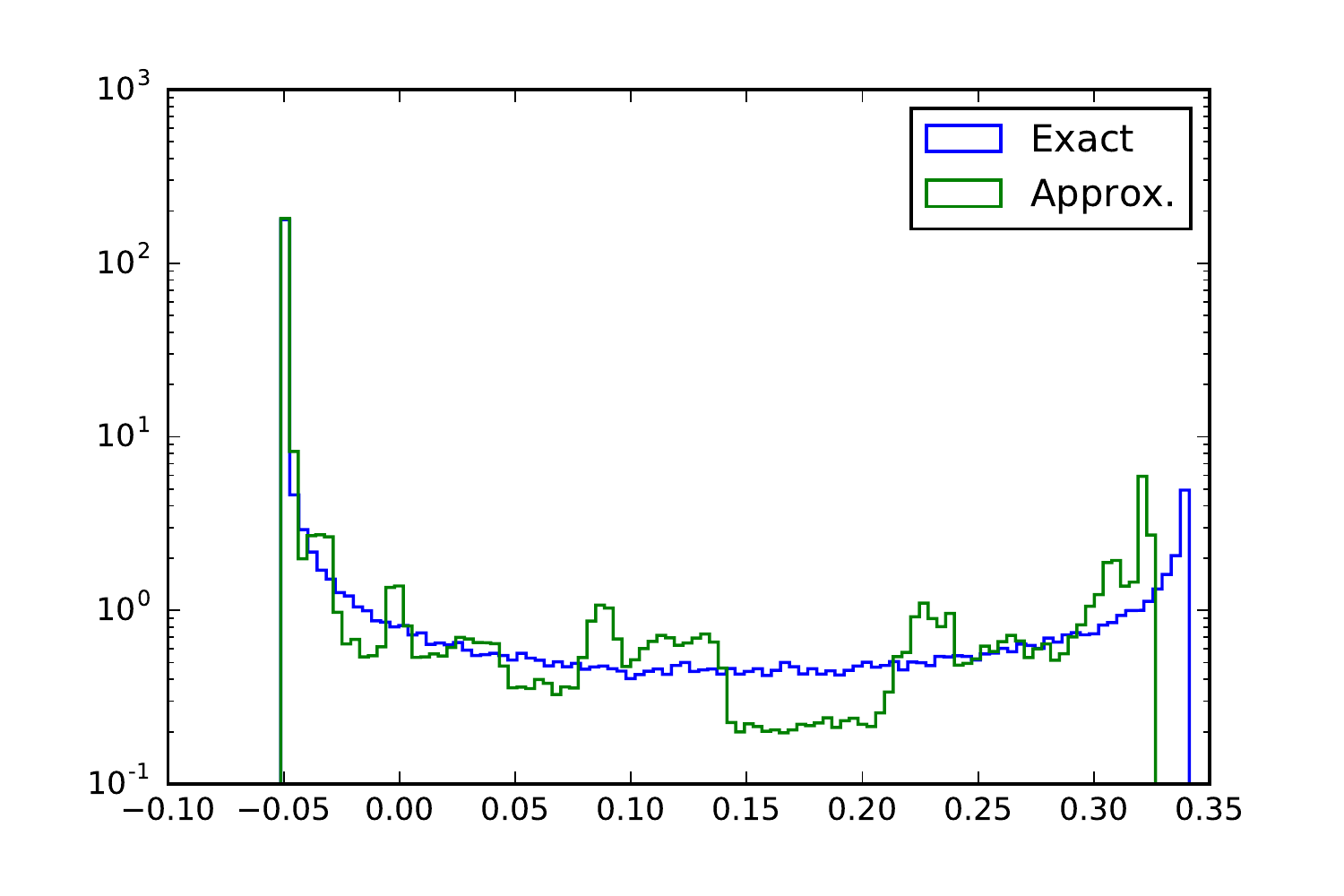}
        \caption{$p( \log \hat r(\hat{s}(\mathbf{x})) \, | \,\gamma = 0.05)$ }
        \label{fig:decomposed_ratio}
    \end{subfigure}
    ~
    \begin{subfigure}[b]{0.4\textwidth}
        \includegraphics[clip, trim=0.5cm 0.5cm 0.5cm 0.5cm,width=\textwidth]{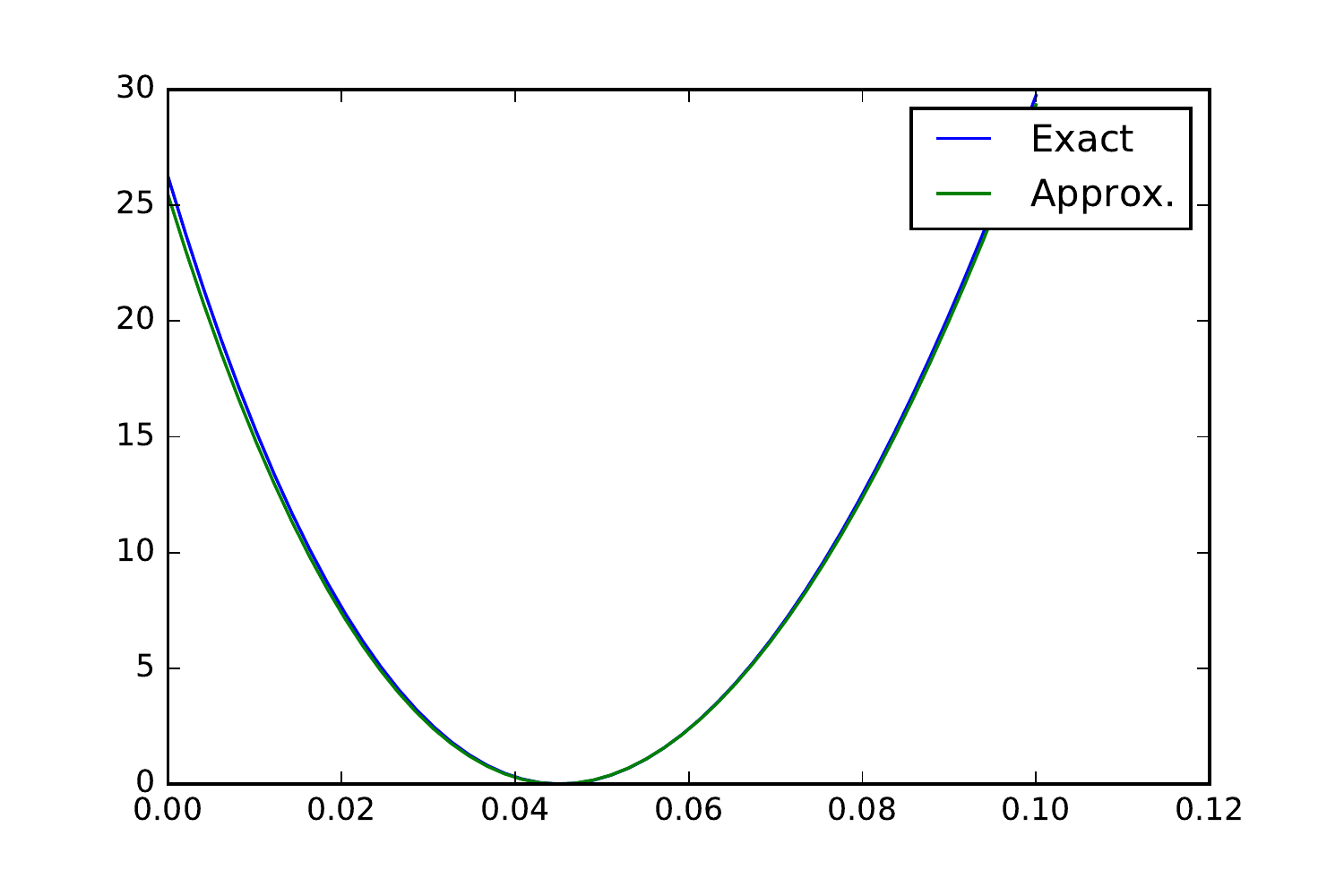}
        \caption{$-2\log \Lambda(\gamma)$ }
        \label{fig:likelihood}
    \end{subfigure}

    \caption{Histogram of $\mathcal{D}$ generated from $\gamma=0.05$ and plots illustrating various stages in 
    the approximation of the log-likelihood ratio  $\log r(\mathbf{x};\gamma=0.05,\gamma=0)$
             with calibrated classifiers (see text). }
\end{figure}

As the results show, calibrating $\hat s(\mathbf{x})$ through univariate density estimation of
$\hat p(\hat s(\mathbf{x}))$ is key to obtaining accurate results. Standard
histograms with uniform binning have been used here for illustrative purposes, but we
anticipate that more sophisticated calibration strategies will be important in further development
of this method. We leave this as an area for future study.

Figure~\ref{fig:decomposed_ratio} shows the distribution of $ \log \hat r(\mathbf{x})$ for $\gamma=0.05$
using the decomposed approximation of $\hat r(\mathbf{x})$ with neural networks and the distribution
of the exact log-likelihood ratio. While there are some artifacts in the distribution in the low-probability regions and the maximum value of the log-likelihood ratio is underestimated, the overall shape of the distribution is well approximated.

Finally, we come to the log-likelihood curve
%shown in Figure~\ref{fig:likelihood}
\begin{equation}
\log \Lambda( \gamma ) =  \log \frac{ p({\cal D} | \gamma ) }{ \sup_{\gamma \in \Theta} p({\cal D} | \gamma )} \;
\end{equation}
for the dataset $\mathcal{D}$ shown in  Fig.~\ref{fig:1a}.
By exploiting Eqn.~\ref{eq:theta_1_independence},
% fact that
%\begin{equation}
%\log \frac{ p({\cal D} | \gamma ) }{ \sup_{\gamma \in \Theta} p({\cal D} | \gamma )} = \log \frac{ p({\cal D} | \gamma ) }{ p({\cal D} | \gamma = 0 ) } - \log \frac{ \sup_{\gamma \in \Theta} p({\cal D} | \gamma) }{  p({\cal D} | \gamma = 0 ) },
%\end{equation}
the generalized likelihood ratio can be computed by evaluating
both terms with respect to a common reference $\gamma=0$
%reusing the decomposed neural network-based approximation from Figure~\ref{fig:1c}.
%In particular, the maximum likelihood estimator (MLE) in the second term
%can be found
as outlined in Sec.~\ref{sec:generalized-likelihood-ratio}.
Figure~\ref{fig:likelihood} shows that the exact likelihood curve is very well
approximated by the method, confirming that
even when the raw classifier does a poor job at modeling the $s^*(\mathbf{x})$, a good approximations of the
likelihood ratio can still be obtained by calibrating $s(\mathbf{x})$ (and by decomposing
the mixture, if possible).

%
%Given observed data from Figure~\ref{fig:1a}, let us now consider the composite hypothesis testing setting, by comparing
%$H_0: \gamma = 0.05$ versus the alternative $H_1: \gamma \neq 0.05$. That is, we want to evaluate
%\begin{equation}
%\Lambda( \gamma = 0.05) =  \frac{ p({\cal D} | \gamma = 0.05) }{ \sup_{\gamma \in \Theta} p({\cal D} | \gamma )} \; .
%\end{equation}
%%\begin{equation}
%%\Lambda({\cal D}; \Theta_0 = \{0.05\}, \Theta=[0, 1]) =  \frac{ p({\cal D} | \gamma = 0.05) }{ \sup_{\gamma \in \Theta} p({\cal D} | \gamma )} \; .
%%\end{equation}
%By exploiting the fact that
%\begin{equation}
%\log \frac{ p({\cal D} | \gamma = 0.05) }{ \sup_{\gamma \in \Theta} p({\cal D} | \gamma )} = \log \frac{ p({\cal D} | \gamma = 0.05) }{ p({\cal D} | \gamma = 0 ) } - \log \frac{ \sup_{\gamma \in \Theta} p({\cal D} | \gamma) }{  p({\cal D} | \gamma = 0 ) },
%\end{equation}
%the generalized likelihood ratio statistic can be computed by evaluating
%both terms reusing the decomposed neural network-based approximation from Figure~\ref{fig:1c}.
%In particular, the maximum likelihood estimator (MLE) in the second term
%can be found as outlined in
%Sec.~\ref{sec:generalized-likelihood-ratio}.

\begin{figure}
    \centering

    \begin{subfigure}[t]{0.475\textwidth}
        \includegraphics[clip, trim=0.5cm 0.5cm 0.5cm 0.5cm,width=\textwidth]{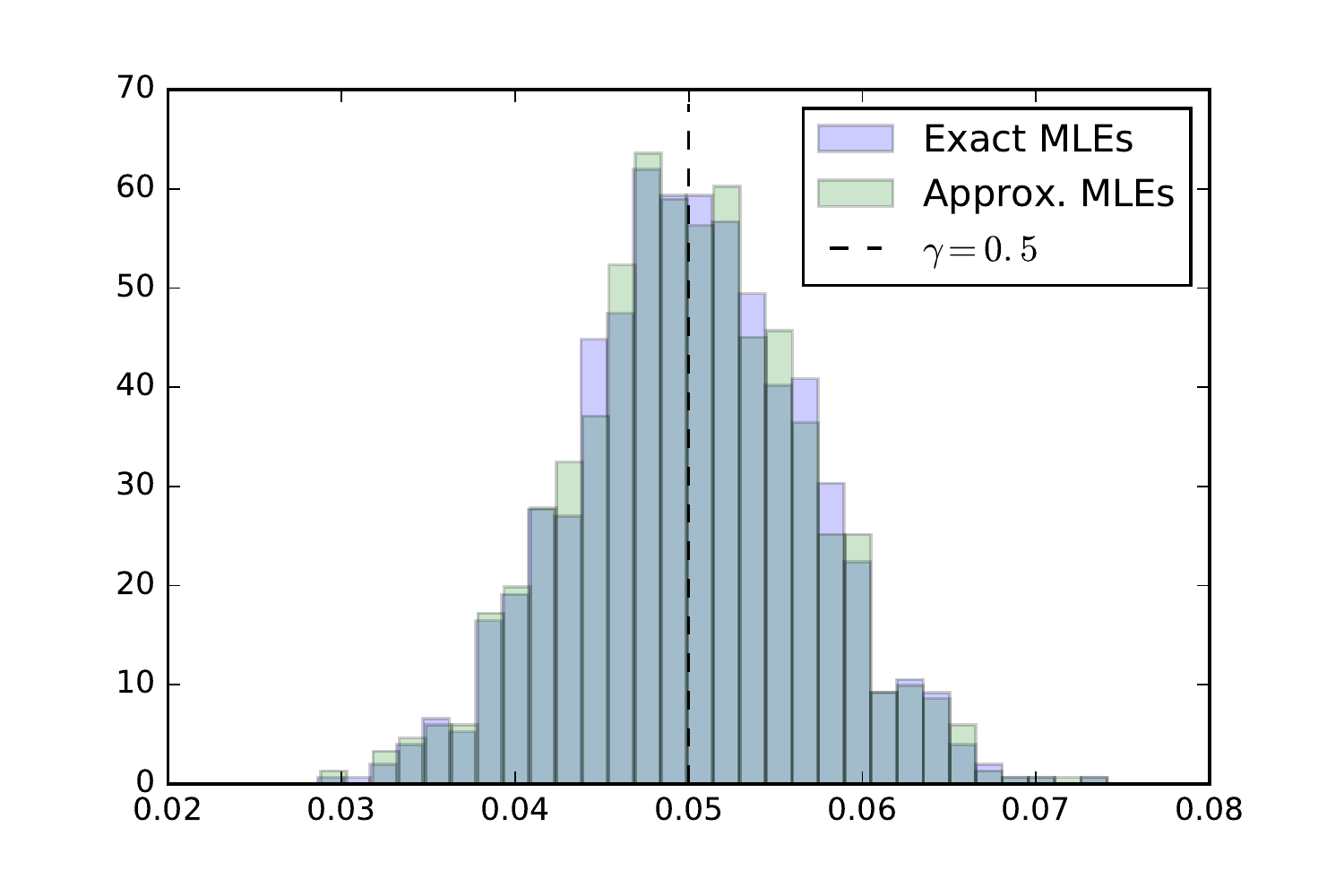}
        \caption{Exact vs. approximated MLEs.}
        \label{fig:2a}
    \end{subfigure}
    ~ %add desired spacing between images, e. g. ~, \quad, \qquad, \hfill etc.
    \begin{subfigure}[t]{0.475\textwidth}
        \includegraphics[clip, trim=0.5cm 0.5cm 0.5cm 0.5cm,width=\textwidth]{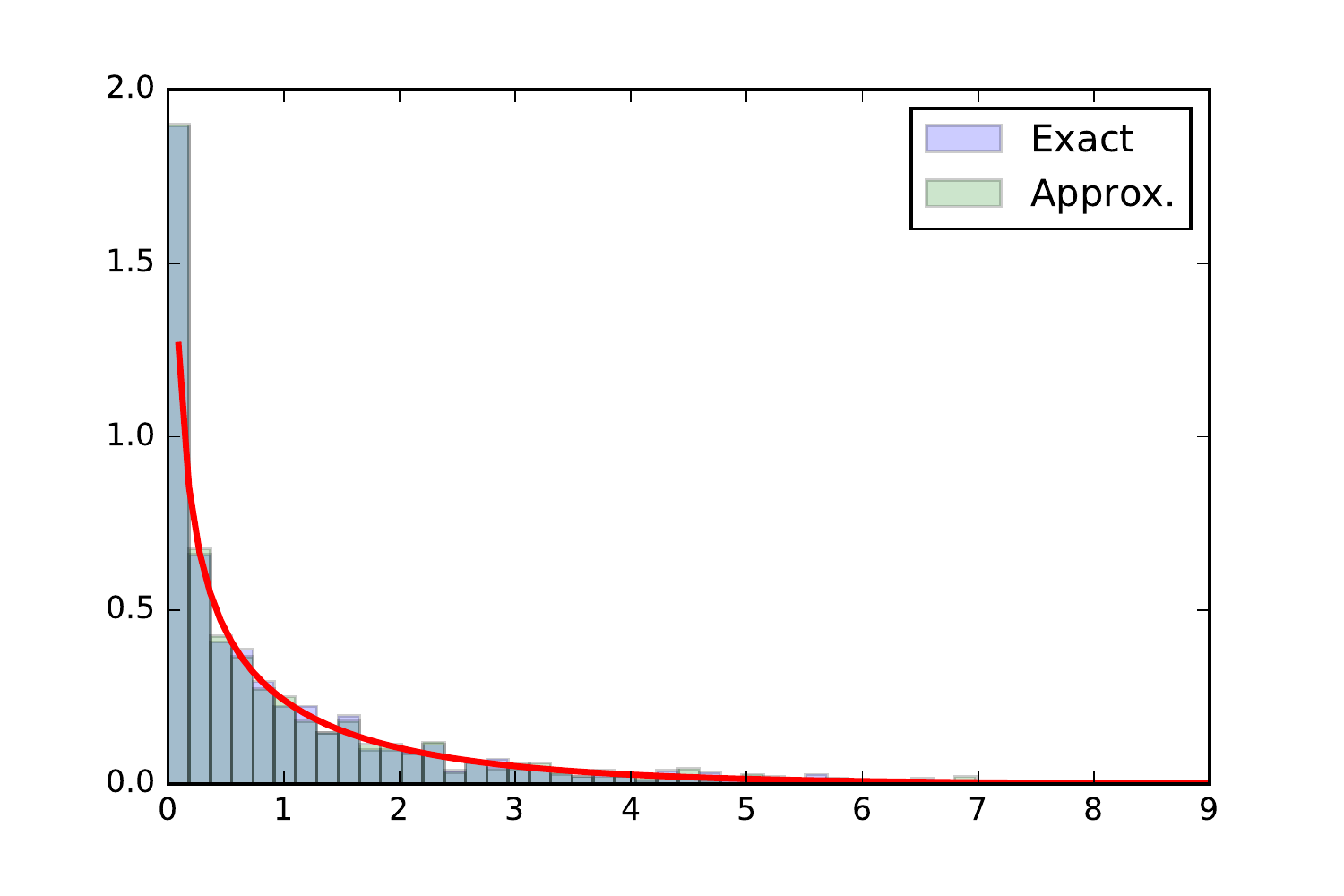}
        \caption{$p(-2 \log \Lambda(\gamma=0.05) \, | \, \gamma=0.05)$}
        \label{fig:2b}
    \end{subfigure}

    \caption{Using approximated likelihood ratios for parameter inference yields an unbiased maximum likelihood estimator
            $\hat \gamma$, as empirically estimated from an ensemble of 1000 artificial datasets.}
    \label{fig:2}
\end{figure}

An advantage of this approach compared to Approximate Bayesian Computation~\citep{beaumont2002approximate} is that the
classifier and calibration -- computationally intensive
parts of the approximation  -- are independent of the dataset $\mathcal{D}$.
Thus once trained and calibrated, the approximation can be applied to any dataset $\mathcal{D}$.
This makes it computationally efficient to perform ensemble tests of the method.

Figure~\ref{fig:2a} shows the empirical distribution of the maximum likelihood estimators (MLEs) from the
approximate likelihood compared to the distribution of the MLEs from the exact likelihood.
It clearly demonstrates that in this case the approximate likelihood yields an unbiased estimator with essentially the same variance as the exact MLE.
%($\mathbb{E} [\hat \gamma] = 0.0501$, as empirically estimated over $1000$
%trials).
%Note that had the underlying ratio not been accurate enough because of
%poor training or poor calibration, the two distributions would not be exactly
%overlapping, indicating bias.
In addition to the MLE, we can study the coverage of a confidence interval based on the likelihood ratio test statistic.
This is done by evaluating $-2\log \Lambda(\gamma=0.05)$ for samples drawn from $p(\mathbf{x} | \gamma=0.05)$.
Wilks's theorem states that the distribution of $-2\log \Lambda(\gamma=0.05)$ should follow a $\chi^2_1$ distribution.
%
%From this, the initial $\Lambda$ test statistic can finally
%be evaluated to determine in this case that the null should not be rejected.
Figure~\ref{fig:2b} also confirms this behavior, supporting the applicability of this method for likelihood-based inference techniques in the likelihood-free setting.
% that $-2 \log \Lambda$, as approximated with our method, is asymptotically $\chi^2_1$ distributed,
%as should be the true statistic when the null is true.

\subsection{Parameterized inference from multidimensional data}

Let us now consider the more challenging problem of likelihood-free
inference with multidimensional data.
For the sake of the illustration, we will assume 5-dimensional feature $\mathbf{x}$ generated
from the following process $p_0$:
    \begin{enumerate}
        \item $\mathbf{z} := (z_0, z_1, z_2, z_3, z_4)$, such that
            $z_0 \sim {\cal N}(\mu=\alpha, \sigma=1)$,
            $z_1 \sim {\cal N}(\mu=\beta, \sigma=3)$,
            $z_2 \sim {\text{Mixture}}(\sfrac{1}{2}\,{\cal N}(\mu=-2, \sigma=1), \sfrac{1}{2}\,{\cal N}(\mu=2, \sigma=0.5))$,
            $z_3 \sim {\text{Exponential}(\lambda=3)}$, and
            $z_4 \sim {\text{Exponential}(\lambda=0.5)}$;
        \item $\mathbf{x} := R  \mathbf{z}$, where $R$ is a fixed
        semi-positive definite $5 \times 5$ matrix defining a fixed projection
        of $\mathbf{z}$ into the observed space.
    \end{enumerate}

\begin{figure}
    \centering
    \includegraphics[width=1.0\textwidth]{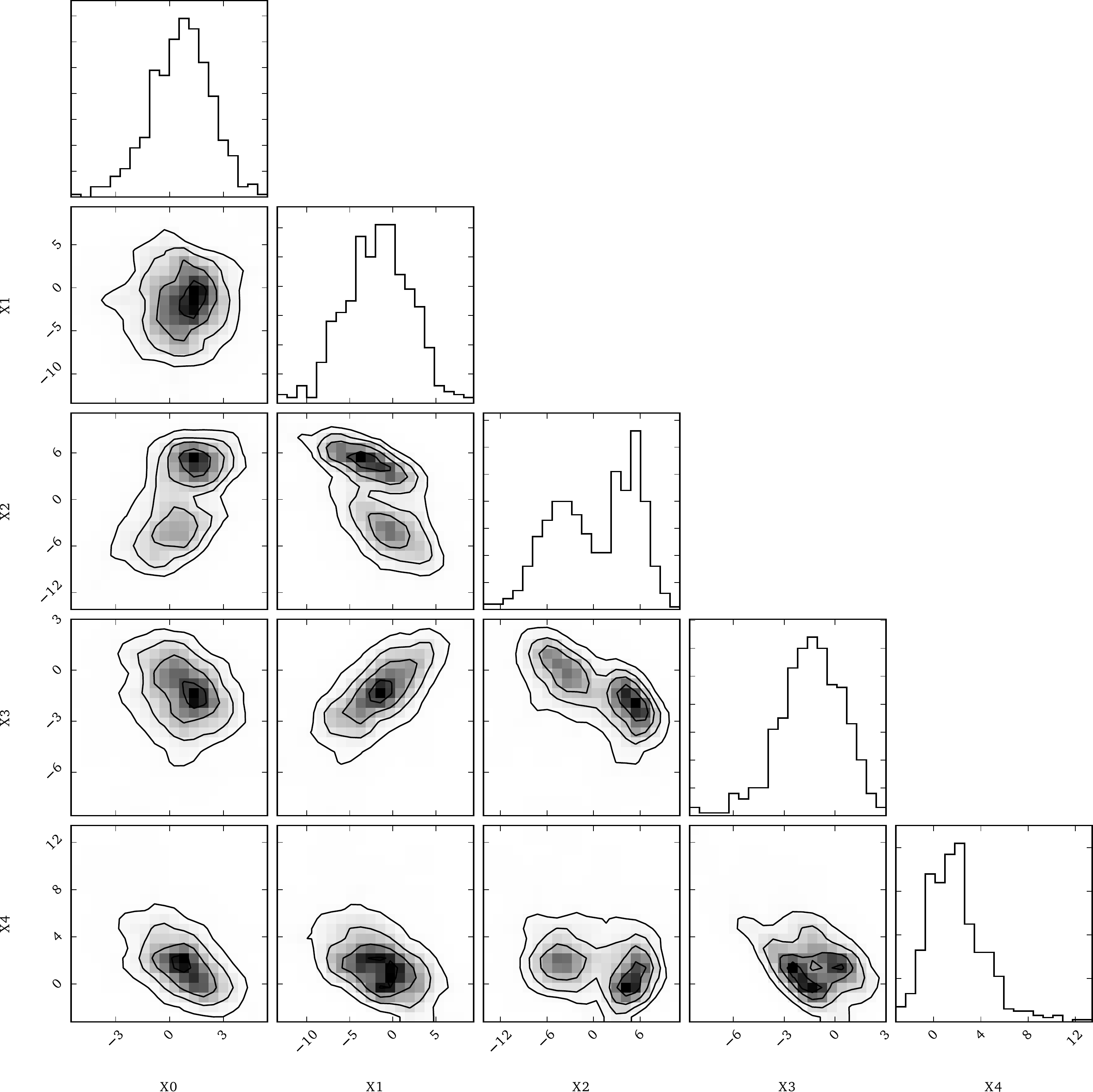}
    \caption{Scatter plots for the 500 samples in the 5-dimensional data ${\cal D}$ generated for nominal values $(\alpha=1, \beta=-1)$. }
    \label{fig:4}
\end{figure}

\begin{figure}
    \centering
    \begin{subfigure}[t]{0.43\textwidth}
        \centering
        \includegraphics[width=\textwidth]{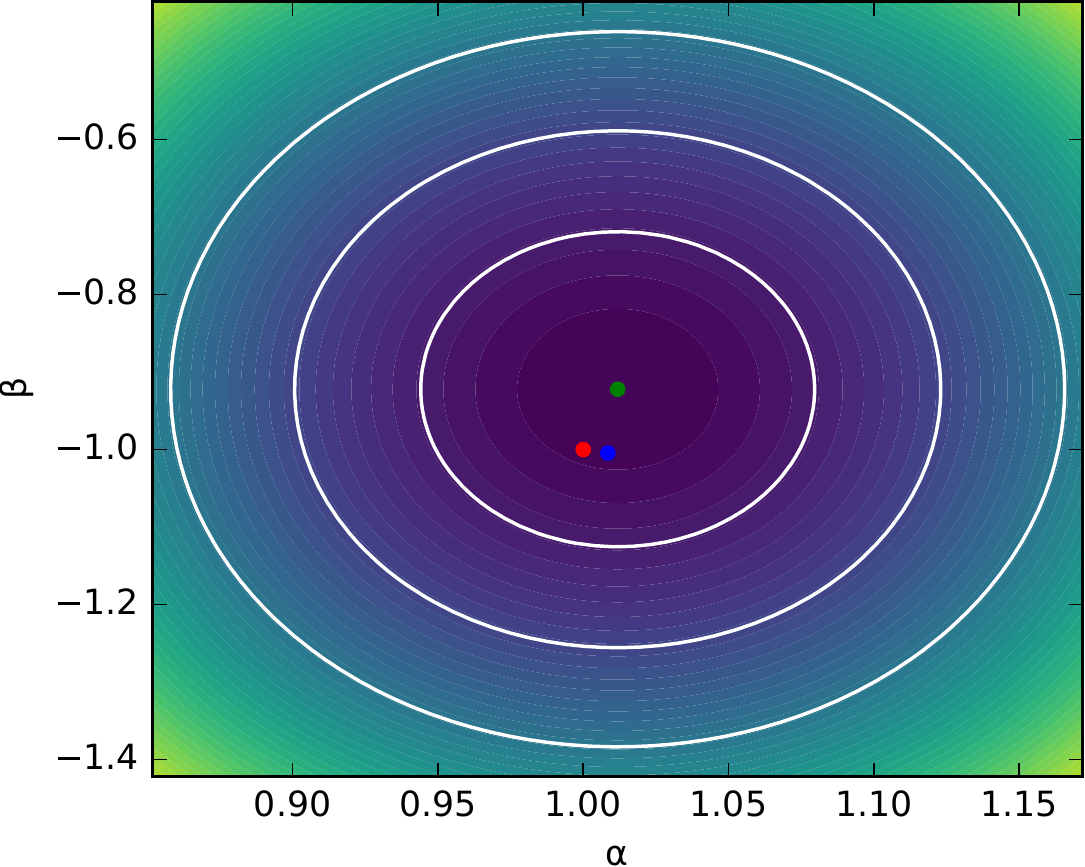}
        \caption{  }
        \label{fig:5a}
    \end{subfigure}
    ~
    \begin{subfigure}[t]{0.43\textwidth}
        \centering
        \includegraphics[width=\textwidth]{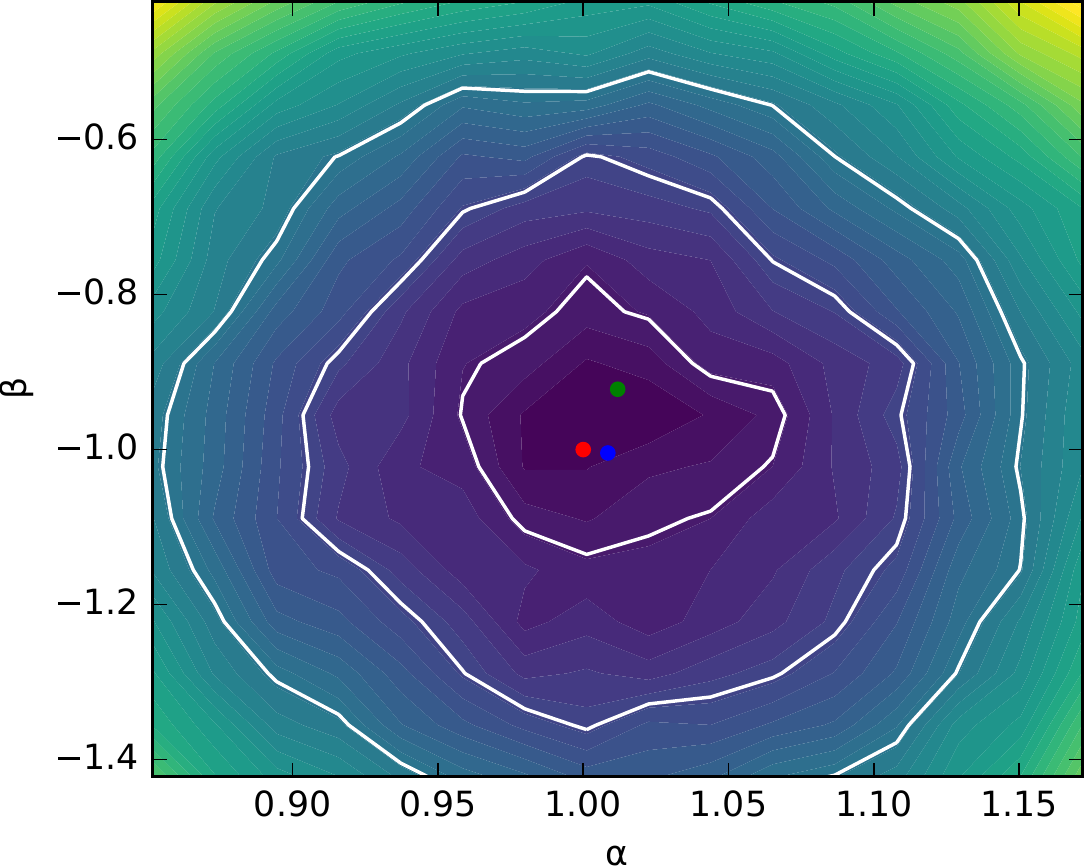}
        \caption{ }
        \label{fig:5b}
    \end{subfigure}

    \begin{subfigure}[t]{0.43\textwidth}
        \centering
        \includegraphics[width=\textwidth]{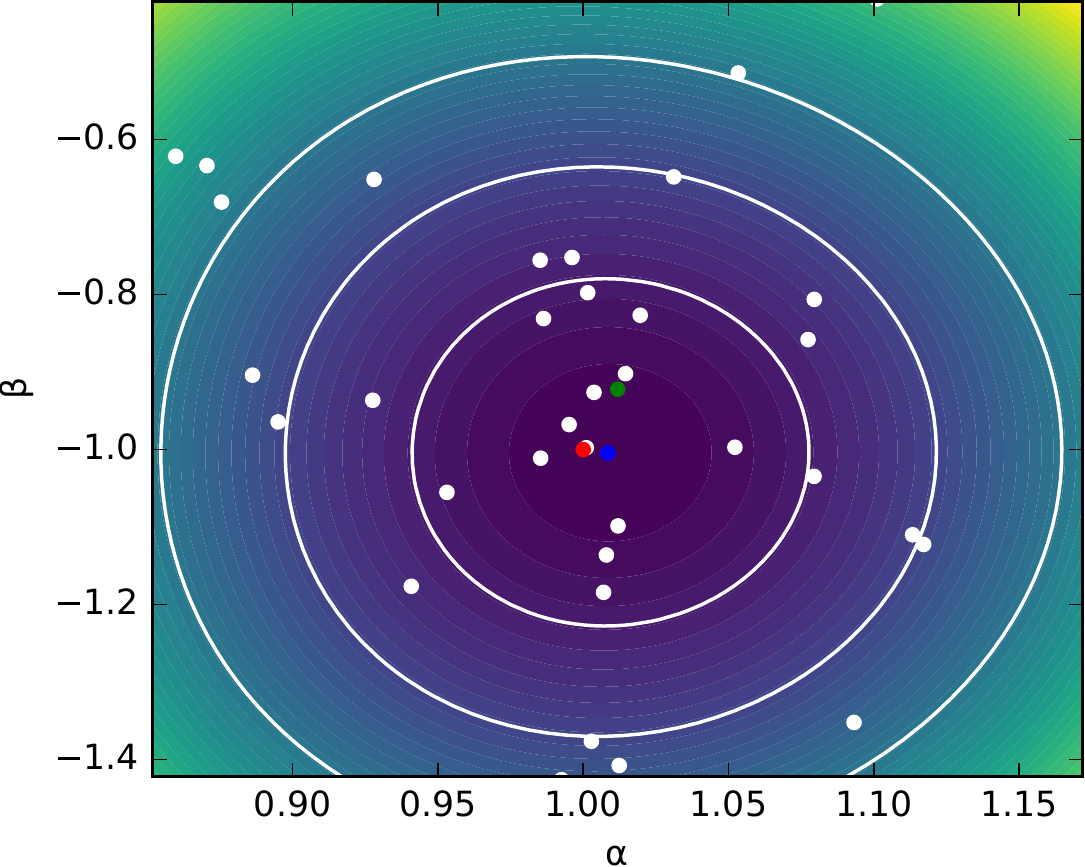}
        \caption{ }
        \label{fig:5c}
    \end{subfigure}
    ~
    \begin{subfigure}[t]{0.43\textwidth}
        \centering
        \hfill\includegraphics[width=0.9\textwidth]{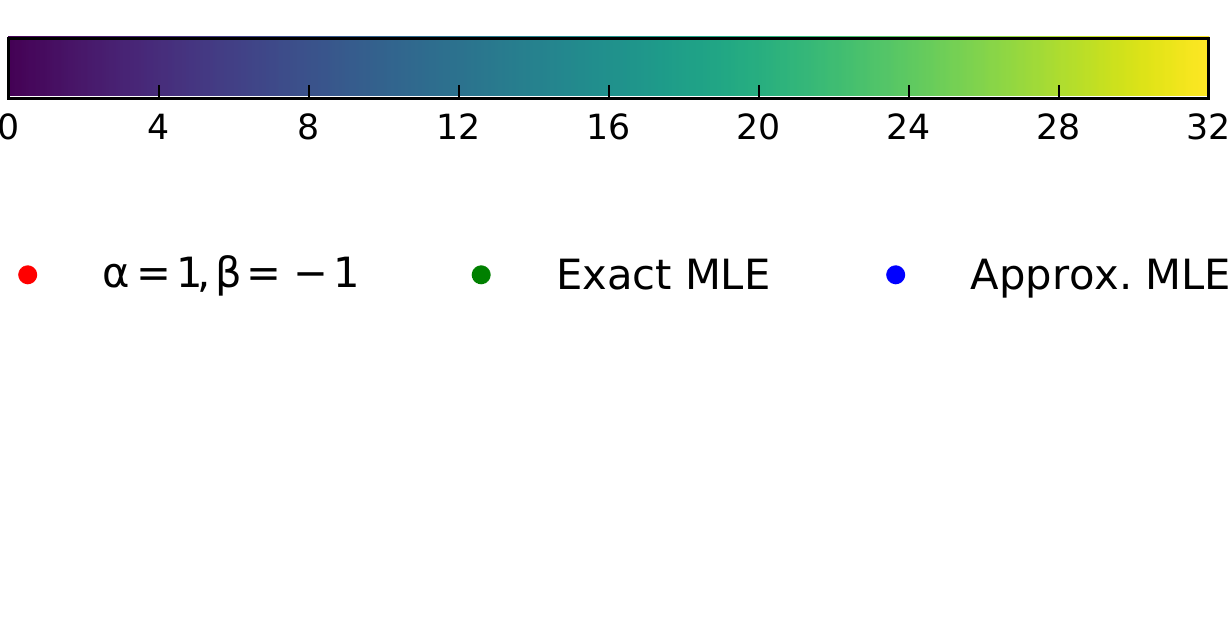}
        \caption{ }
        \label{fig:5d}
    \end{subfigure}
    \caption{Inference from exact and approximate likelihood ratios. The red dot corresponds
    to the true values $(\alpha=1, \beta=-1)$ used to generate $\mathcal{D}$, the green dot is the MLE
     from the exact likelihood, while the blue dot is the MLE  from the approximate likelihood.
    1, 2 and 3-$\sigma$ contours are shown in white.
             (\ref{fig:5a}) The exact $-2\log\Lambda(\alpha,\beta)$ for the observed data ${\cal D}$.
             (\ref{fig:5b}) The approximate $-2\log\Lambda(\alpha,\beta)$ evaluated on a coarse $15 \times 15$  grid.
             (\ref{fig:5c}) A Gaussian Process surrogate of $-2\log\Lambda(\alpha,\beta)$ ratio estimated from a Bayesian optimization procedure. White dots show the  parameter points sampled during the optimization process.}
    \label{fig:5}
\end{figure}

The observations ${\cal D}$ represented in Fig.~\ref{fig:4}
are random samples with $\alpha=1$ and $\beta=-1$.
Our goal is to infer the values $\alpha$ and $\beta$ based on ${\cal D}$.
We construct the log-likelihood ratio
\begin{equation}\label{eqn:obj-ex2}
    -2 \log \Lambda(\alpha, \beta ) = -2 \log \frac{p({\cal D} | \alpha, \beta)}{ \sup_{\alpha,\beta} p({\cal D} | \alpha, \beta) }
\end{equation}
that we calculate via Eqn.~\ref{eq:theta_1_independence}.
%by exploiting the fact that
%\begin{equation}\label{eqn:obj-ex2-expanded}
%\log \frac{p({\cal D} | \alpha, \beta)}{ \sup_{\alpha,\beta} p({\cal D} | \alpha, \beta) } = \log \frac{p({\cal D} | \alpha, \beta)}{ p({\cal D} | \alpha=0, \beta=0) } - \log \frac{ \sup_{\alpha,\beta} p({\cal D} | \alpha, \beta) }{ p({\cal D} | \alpha=0, \beta=0)}.
%\end{equation}
Following the procedure described in Sec.~\ref{sec:param-clf}, we
% to reduce
%computational requirements, both terms in the right hand side of Eqn.~\ref{eqn:obj-ex2-expanded} are evaluated using
build a single 2-layer neural network (with 5+2 inputs and one output node) to form the parameterized classifier $s(\mathbf{x}; \theta_0, \theta_1)$ and fix $\theta_1=(\alpha=0, \beta=0)$.
Since the generative model is not expensive, the classifier output is calibrated
on-the-fly with histograms for  every candidate parameter pair $(\alpha,\beta)$.

Figure~\ref{fig:5a} shows the exact log-likelihood ratio for this dataset, which has an exact MLE at $(\hat \alpha=1.012, \hat \beta=-0.9221)$.
Figure~\ref{fig:5b} shows the approximate log-likelihood ratio evaluated on a coarse grid of parameter values. Some roughness in the contours is observed, which is primarily due to variance introduced in the calibration procedure. In addition to the statistical fluctuations due to finite calibration samples, there are also fluctuations introduced from changes in the binning of the calibration histograms as $\alpha$ and $\beta$ vary. As discussed in Sec.~\ref{sec:param-calibration}, a parameterized calibration
procedure should ameliorate this issue, but that is left for now as an area for future work.
Nevertheless, optimizing the approximate log-likelihood ratio with
a Bayesian optimization~\citep{brochu2010tutorial,gpy2014} procedure is efficient and effective.
After $50$ likelihood evaluations, the maximum likelihood estimate is found at $(\hat \alpha=1.008, \hat \beta=-1.004)$.
While the objective of the Bayesian optimization procedure is to find the maximum likelihood, the posterior mean of the
internal Gaussian process, shown in Fig.~\ref{fig:5c}, is close to the exact log-likelihood ratio illustrated in Fig.~\ref{fig:5a}.
In each case, the true values $\alpha=1$ and $\beta=-1$ are contained within the $1-\sigma$ likelihood contour.
%For comparison, Figure~\ref{fig:5b} shows the approximate log likelihood ratio evaluated on a fixed grid of values. Since $\hat s(\mathbf{x;\alpha,\beta})$ is calibrated on-the-fly,
%roughness in the contours can be observed, in particular due to non-smooth changes in the resulting histograms as $\alpha$ and $\beta$ vary.

\begin{figure}
    \centering
    \begin{subfigure}[t]{0.4\textwidth}
        \centering
        \includegraphics[clip, trim=0.3cm 0.3cm 0.3cm 0.3cm,height=10.075em]{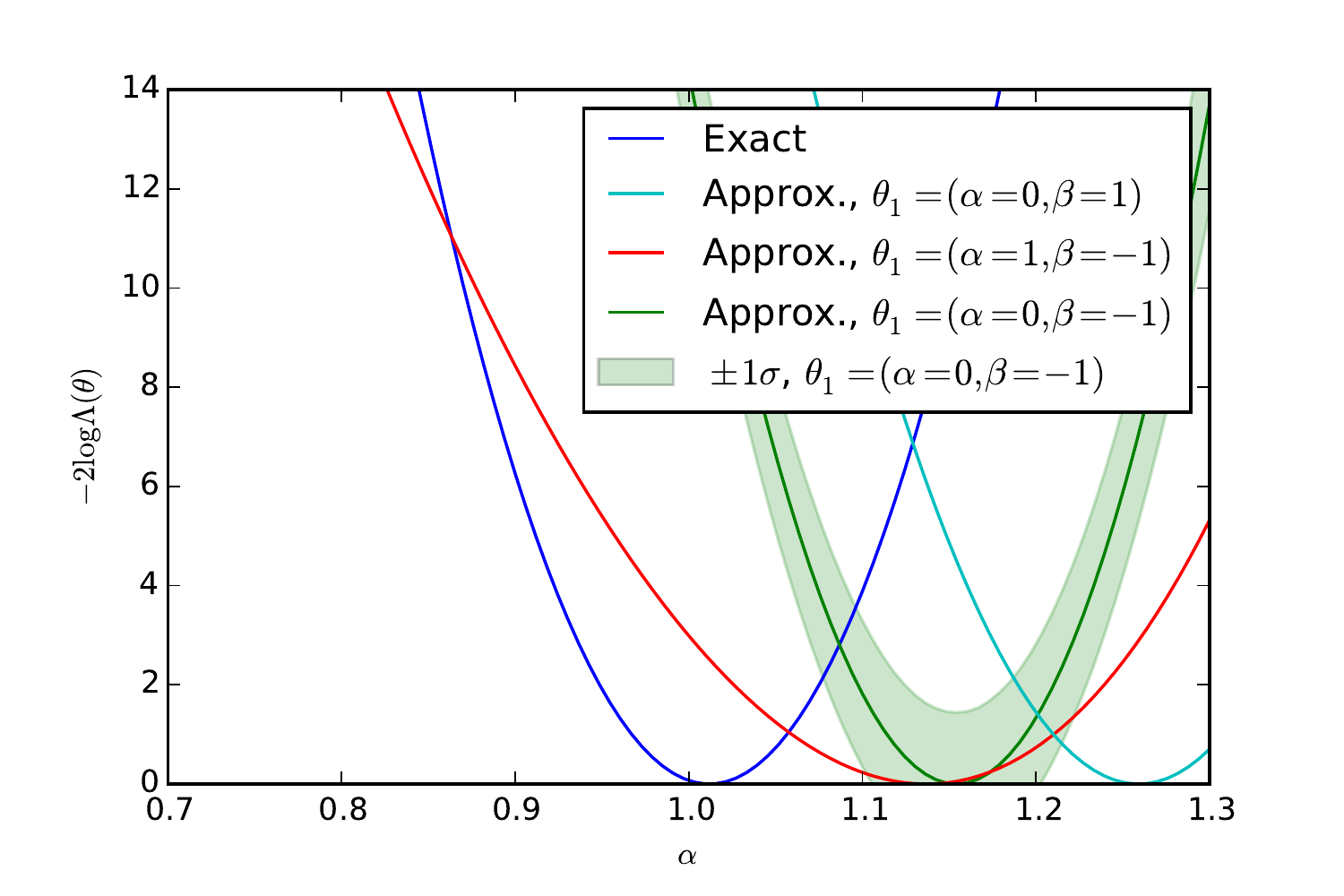}
        \caption{Poorly trained, well calibrated. }
        \label{fig:6a}
    \end{subfigure}
    ~
    \begin{subfigure}[t]{0.4\textwidth}
        \centering
        \includegraphics[clip, trim=0.3cm 0.3cm 0.3cm 0.3cm,height=9.5em]{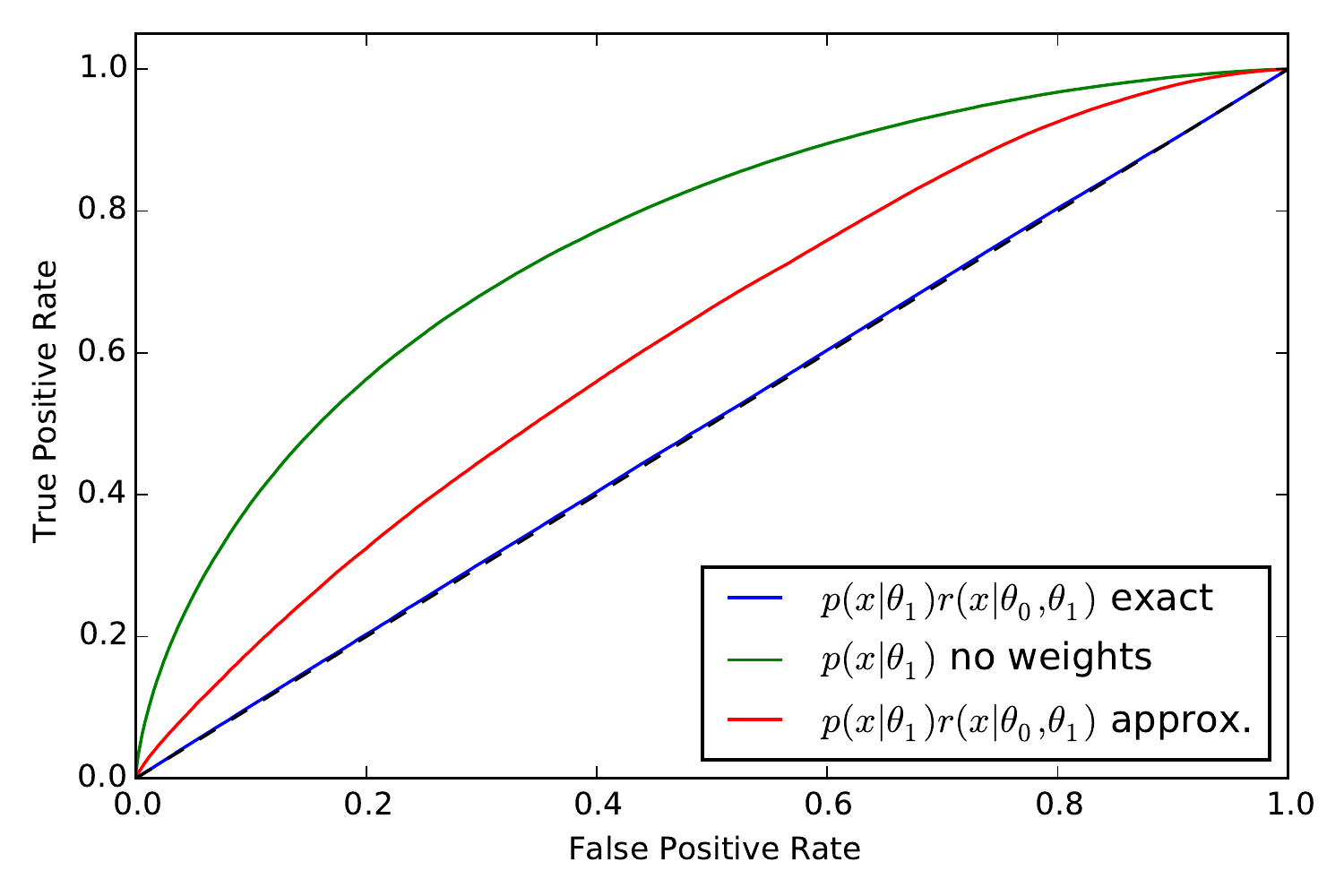}
        \caption{Poorly trained, well calibrated. }
        \label{fig:6b}
    \end{subfigure}

    \begin{subfigure}[t]{0.4\textwidth}
        \centering
        \includegraphics[clip, trim=0.3cm 0.3cm 0.3cm 0.3cm,height=10.075em]{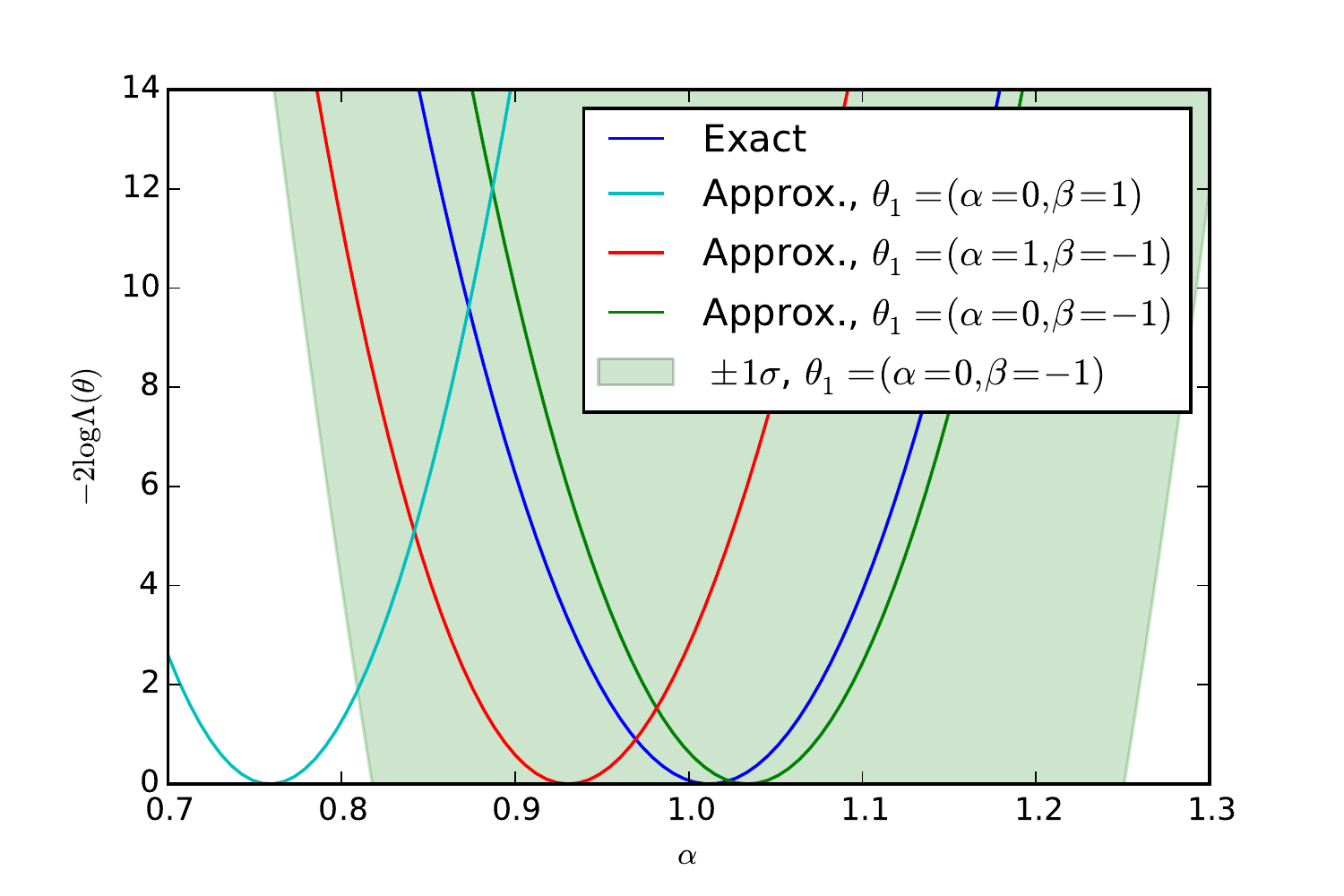}
        \caption{Poorly calibrated, well trained. }
        \label{fig:6c}
    \end{subfigure}
    ~
    \begin{subfigure}[t]{0.4\textwidth}
        \centering
        \includegraphics[clip, trim=0.3cm 0.3cm 0.3cm 0.3cm,height=9.5em]{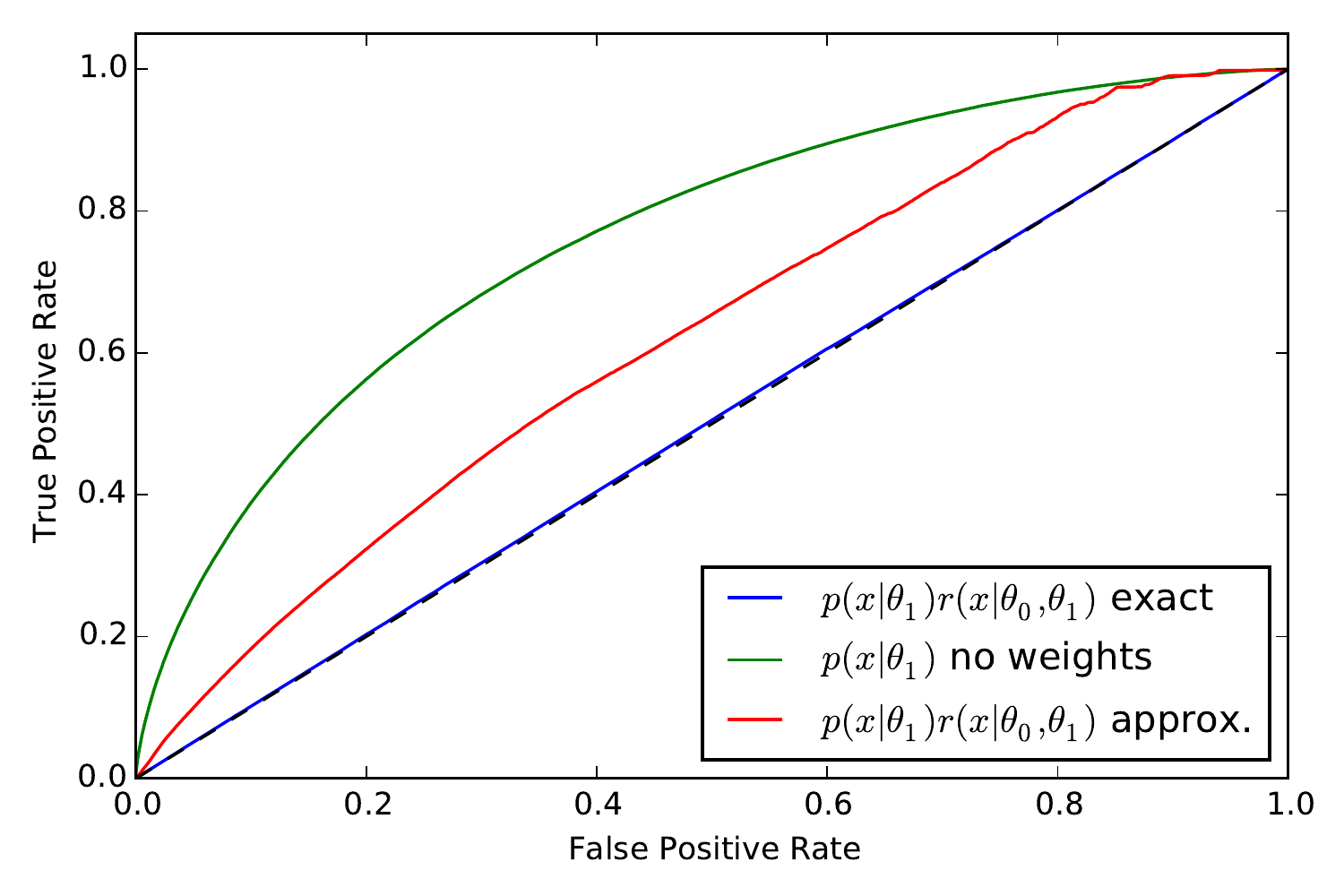}
        \caption{Poorly calibrated, well trained. }
        \label{fig:6d}
    \end{subfigure}

    \begin{subfigure}[t]{0.4\textwidth}
        \centering
        \includegraphics[clip, trim=0.3cm 0.3cm 0.3cm 0.3cm,height=10.075em]{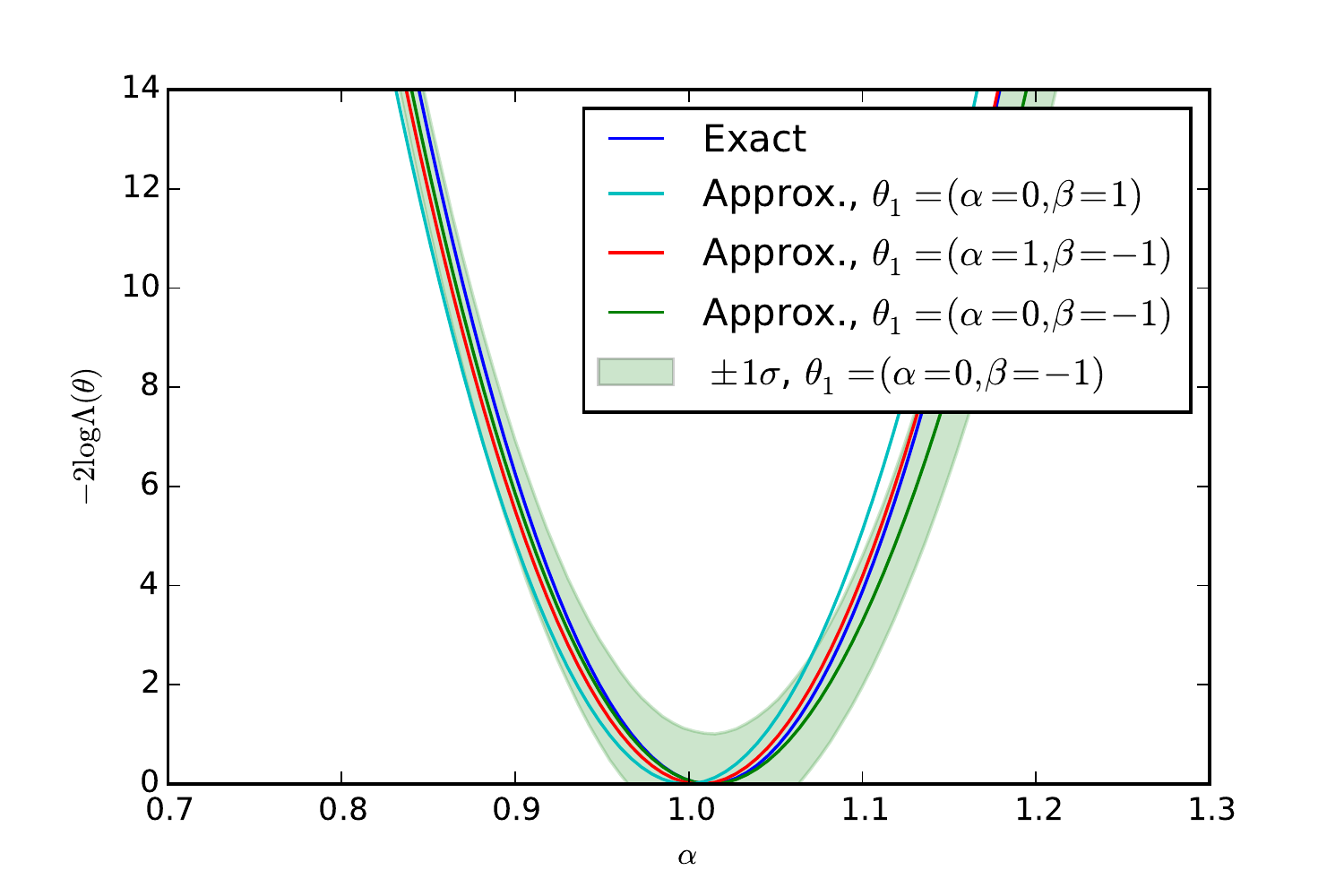}
        \caption{Well trained, well calibrated. }
        \label{fig:6e}
    \end{subfigure}
    ~
    \begin{subfigure}[t]{0.4\textwidth}
        \centering
        \includegraphics[clip, trim=0.3cm 0.3cm 0.3cm 0.3cm,height=9.5em]{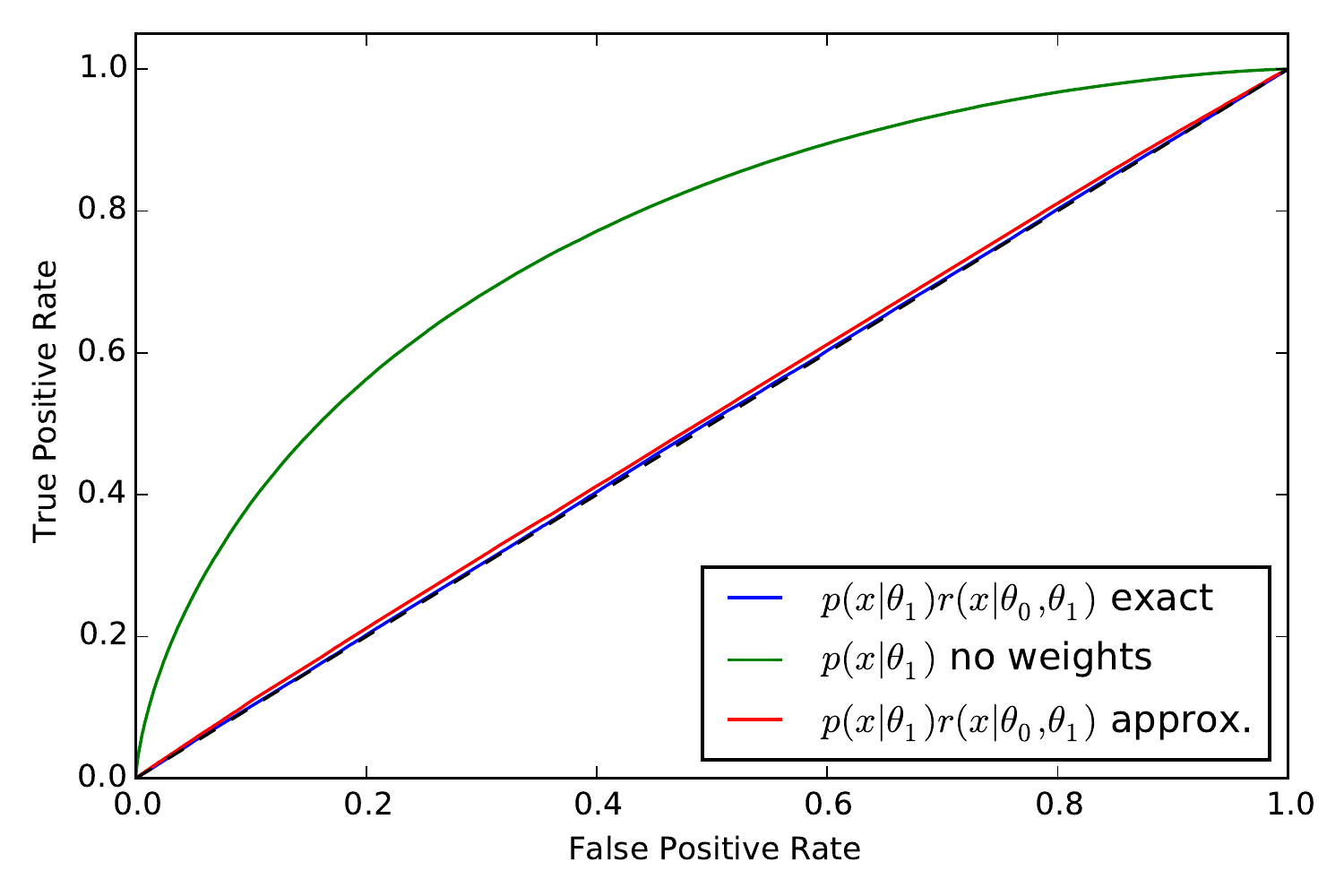}
        \caption{Well trained, well calibrated. }
        \label{fig:6f}
    \end{subfigure}

    \caption{Results from the diagnostics described in Sec.~\ref{S:diagnostics}.
    The rows correspond to the quality of the training and calibration of the classifier. The left plots probe the sensitivity to $\theta_1$, while the right plots show the ROC curve for a calibrator trained to discriminate samples from $p(\mathbf{x}|\theta_0)$ and samples from $p(\mathbf{x}|\theta_1)$ weighted as indicated in the legend. }

%    \caption{Diagnostic plots for the approximate likelihood ratios. The left plots show slices of $-2\log\Lambda(\alpha,\beta=1)$ for the exact likelihood ratio and the approximate likelihood ratio with different choices of $\theta_1$. The right plots show the ROC curve for a calibrator trained to discriminate samples from $p(\mathbf{x}|\theta_0)$ and samples from $p(\mathbf{x}|\theta_1)$ weighted as indicated in the legend. }
    \label{fig:diagnostics}
\end{figure}

Finally, we evaluate the diagnostics described in Sec.~\ref{S:diagnostics} for
this example. To aid in visualization, we restrict to a 1-dimensional slice of
the likelihood along $\alpha$ with $\beta=-1$.   We consider three situations: i)
a poorly trained, but well calibrated classifier; ii) a well trained, but poorly
calibrated classifier; and iii) a well trained, and well calibrated classifier.
For each case, we employ two diagnostic tests. The first checks for independence
of $-2\log\Lambda(\theta)$ with respect to changes in the reference value
$\theta_1$ as shown in Eqn.~\ref{eq:theta_1_independence}. The second uses a
classifier to distinguish between samples from  $p(\mathbf{x}|\theta_0)$ and
samples from $p(\mathbf{x}|\theta_1)$ weighted according to $r(\mathbf{x};
\theta_0, \theta_1)$.  As discussed for Fig.~\ref{fig:5b}, statistical
fluctuations in the calibration lead to some noise in the raw approximate
likelihood. Thus, we show the posterior mean of a Gaussian processes resulting
from Bayesian optimization of  the raw approximate likelihood as in
Fig.~\ref{fig:6c}. In addition, the standard deviation of the Gaussian process
is shown for one of the $\theta_1$ reference points to indicate the size of
these statistical fluctuations. It is clear that in the well calibrated cases
that these fluctuations are small, while in the poorly calibrated case these
fluctuations are large. Moreover, in Fig.~\ref{fig:6a} we see that in the poorly
trained, well calibrated case the classifier $\hat{s}(\mathbf{x}; \theta_0,
\theta_1)$ has a significant dependence on the $\theta_1$ reference point. In
contrast, in Fig.~\ref{fig:6c} the likelihood curves vary significantly, but
this is comparable to the fluctuations expected from the calibration procedure.
Finally, Fig.~\ref{fig:6e} shows that in the well trained, well calibrated case
that the likelihood curves are all consistent with the exact likelihood within
the estimated uncertainty band of the Gaussian process.  The ROC curves tell a
similarly revealing story. As expected, the classifier is not able to
distinguish between the distributions when $p(\mathbf{x}|\theta_1)$ is weighted
by the exact likelihood ratio. We can also rule out that this is a deficiency in
the classifier because the two distributions are well separated when no weights
are applied to $p(\mathbf{x}|\theta_1)$. In both Fig.~\ref{fig:6b} and
Fig.~\ref{fig:6d} the ROC curve correctly diagnoses deficiencies in the
approximate likelihood ratio  $\hat{r}(\hat{s}(\mathbf{x}; \theta_0,
\theta_1))$. Finally, Fig.~\ref{fig:6f} shows that the ROC curve in the well
trained, well calibrated case is almost identical with the exact likelihood
ratio, confirming the quality of the approximation.

Overall, this example further illustrates and confirms the ability
of the proposed method for inference with multiple parameters and multi-dimensional data where
reliable approximations $\hat p(\mathbf{x}|\theta_0)$ and $\hat p(\mathbf{x}|\theta_1)$
are often difficult to construct.

\subsection{High energy physics}

High energy physics was the original scientific domain that motivated the development of this procedure.
In high energy physics, we are often searching for some class of events,
generically referred to as \textit{signal}, in the presence of a separate class
of \textit{background} events.  For each event we measure some quantities
$\mathbf{x}$, with corresponding distributions $p_s(\mathbf{x}|\nu)$ for signal
and $p_b(\mathbf{x}|\nu)$ for background,  where $\nu$ are nuisance
parameters describing uncertainties in the underlying physics prediction or
response of the measurement device. The total model is a mixture of the signal
and background, and $\mu$ is the mixture coefficient associated to the signal
component, that is
\begin{equation}\label{eq:hepGen}
p( {\cal D} | \mu, \nu) = \prod_{\mathbf{x} \in {\cal D}} \left[ \mu p_s( \mathbf{x} |  \nu)  + (1-\mu) p_b( \mathbf{x} | \nu) \right] \;.
\end{equation}
Accordingly, new particle searches are typically framed as
hypothesis tests where the null corresponds to $\mu=0$, and the generalized
likelihood ratio is used as a test statistic.

Nuisance parameters are an after thought in the typical usage of machine
learning in high energy physics. The classifiers are typically trained with
data generated using a fixed nominal value of the nuisance parameters $\nu=\nu_0$.
However, as experimentalists we know that we must account for
the systematic uncertainties that correspond to the nuisance parameters
$\nu$. Thus, typically we take the classifier $\hat s(\mathbf{x})$ as fixed and then propagate
uncertainty by estimating $\hat{p}_s(\hat s(\mathbf{x}) | \nu)$ with a parameterized calibration procedure.
%Building such
%distributions for values of $\nu$ other than the nominal $\nu_0$
%used for training is thought of as a calibration necessary for
%classical statistical inference.
However, this classifier is clearly not optimal
for $\nu \ne \nu_0$. In contrast, a parameterized classifier proposed in this work
would yield more accurate estimates of the generalized likelihood ratio.

% \subsubsection{A more powerful  approach}
%
% The standard use of machine learning in HEP can be improved by training a
% parameterized classifier corresponding to the generalized likelihood ratio test
% \begin{equation}
% \lambda(\mu) = \frac{p(D|\mu, \hat{\hat{\nu}})}{p(D|\hat \mu, {\hat{\nu}})} \;,
% \end{equation}
% following the approach outlined in Sec.~\ref{S:GLR}.
%
% There is an interesting distinction between this approach and the standard use
% in which the classifier is trained for a fixed $\nu_0$. In the standard use one
% trains a classifier for signal vs. background, which is equivalent (in an ideal
% setting) to training a classifier for  null (background-only) vs. alternate
% (signal-plus-background) since the resulting regression functions are one-to-one
% with each other. In contrast, in the case of the generalized likelihood ratio
% test
% \begin{equation}\label{eq:hep_improved}
%  \frac{p(x| 0, \hat{\hat{ \nu}})}{p(x|\hat \mu, \hat\nu)} =  \frac{p_b(x| \hat{\hat{ \nu}})}{ \hat \mu p_s( x_e \, |\,  \hat\nu)  + (1- \hat \mu )\, p_b( x_e \,|\, \hat \nu)} \; ,
% \end{equation}
% the background components don't cancel and there is an additional term $p_b(x|
% \hat{\hat{ \nu}})/p_b(x| {\hat{ \nu}})$. In practice, with classifiers of finite
% capacity, there will be some trade-off between taking into account this
% additional term and the more challenging learning problem when $\mu$ is very
% small.

In addition to robustness to systematic uncertainties incorporated by the nuisance parameters $\nu$,
the proposed method can be used to infer parameters of interest. Not only can the mixture coefficient $\mu$
be inferred using the decomposition procedure, but also physical parameters like particle masses that change the distribution of $\mathbf{x}$.
%While the original motivation for this work was to improve the treatment of
%systematic uncertainties in new particle searches by parameterizing the
%classifier in terms of the nuisance parameters $\nu$, the same approach can be
%used for parameter inference. In the case of new particle searches, the parameter
%of interest is the mixture coefficient for the signal component $p_s(x|\nu)$.
%When measuring particle properties, the distribution of the features also depend
%on parameters such as a particle's mass and quantum numbers. This is easily
%accommodated by extending $p_s(\mathbf{x}|\nu) \to p_s(\mathbf{x}|\theta)$, where $\theta$
%includes both parameters of interest and nuisance parameters.
This formalism represents a significant step forward in the usage of machine
learning in high energy physics, where classifiers have always been used between two static
classes of events and not parameterized explicitly in terms of the physical
quantities we wish to measure.

%The work of \cite{Whiteson:2006ws} used a stochastic optimization technique
%to find a binary classifier such that selected events would minimize the measurement
%uncertainty of a particle's mass; however, the resulting classifier was fixed.
%This approach also offers the advantage that it explicitly reformulates the
%per-experiment optimization to the per-event optimization, which is less
%computationally intensive.

%The work of \cite{Whiteson:2006ws} is similar,
%as the stochastic optimization was directly trying to minimize the measurement
%uncertainty of a particle's mass; however, the resulting classifier was fixed.
%This approach also offers the advantage that it explicitly reformulates the
%per-experiment optimization to the per-event optimization, which is less
%computationally intensive.

Another approach for parameter inference with multi-dimensional data specific to
high energy physics is the so-called matrix element
method, in which one  directly computes an approximate likelihood ratio by
performing a computationally intensive integral associated to a simplified detector
response~\citep{Volobouev:2011vb}. In the approach considered in this paper, the
detailed detector response is naturally incorporated by the simulator; however, that integral
is intractable for the matrix element method. Even with drastic simplifications of the detector response, the
matrix element method can take several minutes of CPU time to calculate the
likelihood ratio for a single event $\mathbf{x}$. The work here can be seen as aiming
at the same conceptual target, but relying on machine learning to overcome the
complexity of the detector simulation. It also offers enormous speed increase
for evaluating the likelihood at the cost of an initial training stage. In
practice, the matrix element method has only been used for searches and
measurement of a single physical parameter (sometimes with a single nuisance
parameter as in~\citep{Aaltonen:2010yz}).

Contemporary examples where the technique presented here could have major impact
include the measurement of coefficients to quantum mechanical operators
describing the production and decay of the Higgs boson~\citep{Chen:2014pia} and, if we are so
lucky, measurement of the mass of supersymmetric particles in cascade
decays~\citep{Allanach:2000kt}.  Both of these examples involve data sets with
many events, each with a feature vector $\mathbf{x}$ that has on the order of 10
components, and a parameter vector $\theta$ with 2-10 parameters of interest and
possibly many more nuisance parameters.
%The state of the art for the operator
%coefficients of the Higgs decay uses the so-called matrix element likelihood
%analysis (MELA) in which the equivalent of $s(\mathbf{x}; \theta_0, \theta_1)$ is
%approximated by neglecting detector effects~\citep{Gao:2010qx,Bolognesi:2012mm}.

% Related works ================================================================

\section{Related work}
\label{sec:related}

The closest work to the proposed method is due to \cite{Neal:2007zz}, who
similarly considers the problem of approximating the likelihood function when
only a generative model is available. That work sketches a scheme in which one
uses a classifier with both $\mathbf{x}$ and $\theta$ as an input to serve as a
dimensionality reduction map. The key distinction comes in the handling of
$\theta$.  Neal argues that a classifier cannot be used on real data, since we
do not know the correct value for $\theta$, and goes on to outline an approach
where one uses regression on a per-event basis to estimate
$\hat{\theta}(\mathbf{x})$ and perform the composition $s(\mathbf{x};
\hat{\theta}(\mathbf{x}))$. As pointed out by the author, this can lead to a
significant loss of information since a single observation $\mathbf{x}$ may
carry little information about the true value of $\theta$, though a full data
set ${\cal D}$ may be informative.
% for instance, a single observation would
% not be sufficient to estimate the variance of a distribution, though repeated
% observations would.
The work of \cite{Neal:2007zz} correctly identifies this as
an approximation of the target likelihood even in the case of a ideal
classifier. In contrast, the approach described here does not eliminate the
dependence of the classifier on $\theta$.
% \footnote{As a technical point, in
% \citep{Neal:2007zz}, the focus is on approximating the likelihood function (up
% to a multiplicative constant), which is equivalent to evaluating the ratio  with
% respect to a fixed $\theta_1$ as in Eqn.~\ref{eq:mle_withs}. In
% \citep{Neal:2007zz}, the dependence on $\theta$ is eliminated via $s(\mathbf{x};
% \hat{\theta}(\mathbf{x}))$ and the map is constant; however, in this approach
% the map ratio is explicitly parameterized in terms of $\theta$, so the ratio is
% important for canceling the corresponding Jacobian factors. \glnote{Shall we
% keep this footenote? Would be better to give details in the text. JASA
% recommends not using footnotes.}}.
Instead, we embed a parameterized classifier
into the likelihood and postpone the evaluation of the classifier to the point
of evaluation of the likelihood when $\theta$ is explicitly being tested. This
avoids the loss of information that occurs from the regression step
$\hat{\theta}(\mathbf{x})$ proposed by \cite{Neal:2007zz} and leads to
Thm.~\ref{thm:ratio-equivalence}, which is an exact result in the case of an
ideal classifier. In both cases, the quality of the classifier is factorized
from the calibration of its density, which allows for valid inference even if
there is a loss of power due to a non ideal classifier.

Also close to our work, \cite{ClaytonScott} and \cite{JMLR:v14:tong13a} consider
the machine learning problem associated to Neyman-Pearson hypothesis testing. In
a similar setup, they consider the situation where one does not have access to
the underlying distributions, but only has i.i.d. samples from each hypothesis.
This work generalizes that goal from the Neyman-Pearson setting to generalized
likelihood ratio tests and emphasizes the connection with classification.
% Perhaps a  formal treatment similar to the Neyman-Pearson case can be brought to
% bear in this more general setting.
\cite{Ihler2004} take on a different problem
(tests of statistical independence) by using machine learning algorithms to find
scalar maps from the high-dimensional feature space that achieve the desired
statistical goal when the fundamental high-dimensional test is intractable.

% In a similarly titled work, \cite{Gutmann2014} advocate using
% the cross-validated classification accuracy as the similarity metric used in
% ABC. While the goal there is also parameter inference in the likelihood-free
% setting,  their method is very different than the approach presented here.
% \cite{TommiJaakkola} explore a way of leveraging generative models to derive
% kernel functions for use in discriminative methods. This interesting work is
% distinct from the point made here in which the generative model is being used
% for the purpose of providing training data and calibration.
% \cite{McCallum} consider a hybrid generative/discriminative classifier;
% however, the goal of that work is not to leverage a generative model for the
% data, but to use both approaches to learn different subsets of the parameters in
% a single hybrid classifier.
% \cite{BiancaZadrozny} emphasize the importance of calibrated probability
% estimates from decision trees and naive Bayesian classifiers and investigate
% various approaches to achieve this. In contrast to that work, we are not
% interested in calibrated probability estimates for $p(y|x)$ for individual
% events, but instead we use the calibration to correct for non-linear
% transformations of the target likelihood ratio and, perhaps, to provide
% calibrated p-values based on those likelihood ratio tests.

More generally, likelihood ratio testing directly relates to the density ratio
estimation problem, which consists in estimating the ratio of two
densities from finite collections of observations ${\cal D}_0$ and
${\cal D}_1$. Density ratio estimation is connected to many machine learning
fundamental problems, including transfer learning~\citep{sugiyama2012machine},
probabilistic classification and regression~\citep{vapnik1998statistical},
outlier detection~\citep{hido2011statistical}, and many others. For learning
under covariate shift, \cite{shimodaira2000improving} and \cite{sugiyama2005input} estimate
the density ratio $r(\mathbf{x};\theta_0,\theta_1)$ from straightforward
approximations $\hat{p}(\mathbf{x}|\theta_0)$ and $\hat{p}(\mathbf{x}|\theta_1)$
separately obtained using kernel density estimation. Despite its theoretical
consistency, this approach is known to be ineffective in
practice~\citep{sugiyama2007covariate,bickel2009discriminative}, since
it relies on modeling numerator and denominator high-dimensional densities,
which is a harder problem than modeling their ratio only.
While the proposed method also proceeds in two similar steps, estimating
$p(s(\mathbf{x}))$ is much easier than estimating $p(\mathbf{x})$,
since $s$ projects $\mathbf{x}$ into a one-dimensional space in which only
the informative content of $r(\mathbf{x})$ is preserved.
Finally, in contrast with the
proposed method which decouples reduction from calibration, other
approaches proposed within the literature (see
\cite{sugiyama2012density,gretton2009covariate,nguyen2010estimating,vapnik2013constructive}
and references therein) provide solutions for estimating
$r(\mathbf{x};\theta_0,\theta_1)$ directly from $\mathbf{x}$, in one step. Under some
assumptions, the convergence of the obtained estimates is also proven for some
of these approaches.

% Conclusions ==================================================================

\section{Conclusions}
\label{sec:conclusions}

In this work, we have outlined an approach to reformulate generalized likelihood
ratio tests with a high-dimensional data set in terms of univariate densities
of a classifier score. We have shown that a parameterized family of
discriminative classifiers $\hat s(\mathbf{x}; \theta_0, \theta_1)$ trained and
calibrated with a simulator can be used to approximate the likelihood ratio,
even when it is not possible to directly evaluate the likelihood
$p(\mathbf{x}|\theta)$.
% A technique for decomposing this ratio when the generative model is a mixture
% of components was presented with the aim to help focus capacity of the
% classifier when $p(x|\theta_0)$ and $p(x|\theta_1)$ differ primarily by a small
% mixture coefficient.  This approach leverages the power of machine learning in a
% classical statistical setting.
The proposed method offers an alternative to Approximate Bayesian Computation
for parameter inference in the likelihood-free setting that can also be used in
the frequentist formalism without specifying a prior over the parameters. 
In contrast to approaches that learn the posterior conditional on $\mathcal{D}$, our
approach can be applied to any observed data $\mathcal{D}$ once trained.
A strength of this approach is that it separates the quality of the approximation
of the target likelihood from the quality of the calibration. The former
leverages the continuing advances in supervised learning approaches to classification. 
The calibration procedure for a particular
parameter point is fairly straightforward since it involves estimating a
univariate density using a generative model of the data. The difficulty of the
calibration stage is performing this calibration continuously in $\theta$.
Different strategies to this calibration are anticipated depending on the
dimensionality of $\theta$, the complexity of the resulting likelihood function,
or the practical issues associated to running the simulator.

\section*{Acknowledgments} KC and GL are both supported through NSF ACI-1450310, 
additionally KC is supported through PHY-1505463 and PHY-1205376.
JP was partially supported by the Scientific and Technological Center of Valpara\'iso (CCTVal) under Fondecyt grant BASAL FB0821.
KC would like to thank Daniel Whiteson for
encouragement and Alex Ihler for challenging discussions
that  led to a reformulation of the initial idea. KC would also like to thank
Shimon Whiteson, Babak Shahbaba for advice in the presentation, Radford Neal for
discussion of his earlier work, Yann LeCun, Philip Stark, and Pierre Baldi for
their feedback on the project early in its conception, Bal\'azs K\'egl for
feedback to the draft,  and Yuri
Shirman for reassuring cross checks of the Theorem.   KC is grateful
to UC-Irvine for their hospitality while this research was carried out and the
Moore and Sloan foundations for their generous support of the data science
environment at NYU.

\bibliographystyle{apalike}
\bibliography{learning.bib}

\appendix

\section{Probabilistic classification for building $s$}
\label{app:clf-for-s}

In this appendix, we show for completeness that the probabilistic classification framework
yields a reduction $s$ which satisfies conditions of Thm.~\ref{thm:ratio-equivalence}.

\begin{proposition} \label{thm:best-classifier}
Let $\mathbf{X} = (X_1, ..., X_p)$ and $Y$ be random input and output variables
with values in ${\cal X} \subseteq \mathbb{R}^p$
and ${\cal Y} = \{0, 1\}$ and mixed joint probability density  function
$p_{\mathbf{X},Y}(\mathbf{x}, y)$. For the squared error loss, the best
regression function $s : {\cal X} \mapsto [0, 1]$, or equivalently the best
probabilistic classifier, is
\begin{equation}
s^*(\mathbf{x}) = \frac{P(Y=1) p_{\mathbf{X}|Y}(\mathbf{x}|Y=1)}{P(Y=0) p_{\mathbf{X}|Y}(\mathbf{x} | Y=0) + P(Y=1) p_{\mathbf{X}|Y}(\mathbf{x} | Y=1)}.
\end{equation}
\end{proposition}

\begin{proof}
For the squared error loss,
\begin{align}
s^*(\mathbf{x}) &= \argmin_{s(\mathbf{x})} \mathbb{E}_{Y|\mathbf{X}=\mathbf{x}} \{ (Y - s(\mathbf{x}))^2 \} \nonumber \\
&=  \argmin_{s(\mathbf{x})} \mathbb{E}_{Y|\mathbf{X}=\mathbf{x}} \{ Y^2 \} - 2s(\mathbf{x}) \mathbb{E}_{Y|\mathbf{X}=\mathbf{x}} \{ Y \} + s(\mathbf{x})^2 \nonumber \\
&=  \argmin_{s(\mathbf{x})} -2s(\mathbf{x}) \mathbb{E}_{Y|\mathbf{X}=\mathbf{x}} \{ Y \} + s(\mathbf{x})^2
\end{align}
The last expression is minimized when $\frac{d}{ds(\mathbf{x})} (-2s(\mathbf{x}) \mathbb{E}_{Y|\mathbf{X}=\mathbf{x}} \{ Y \} + s(\mathbf{x})^2) = 0$,
that is when $-2 \mathbb{E}_{Y|\mathbf{X}=\mathbf{x}} \{ Y \} + 2 s(\mathbf{x}) = 0$, hence
\begin{equation}
s^*(\mathbf{x}) = \mathbb{E}_{Y|\mathbf{X}=\mathbf{x}} \{ Y \}.
\end{equation}
For ${\cal Y} = \{ 0, 1 \}$,
\begin{align}
\mathbb{E}_{Y|\mathbf{X}=\mathbf{x}} \{ Y \} &= P(Y=0|\mathbf{X}=\mathbf{x}) \times 0 +  P(Y=1|\mathbf{X}=\mathbf{x}) \times 1 \nonumber \\
&= \frac{P(Y=1) p_{\mathbf{X}|Y}(\mathbf{x}|Y=1)}{p_{\mathbf{X}}(\mathbf{x})} \nonumber \\
&= \frac{P(Y=1) p_{\mathbf{X}|Y}(\mathbf{x}|Y=1)}{P(Y=0) p_{\mathbf{X}|Y}(\mathbf{x} | Y=0) + P(Y=1) p_{\mathbf{X}|Y}(\mathbf{x} | Y=1)}.
\end{align}
\end{proof}

For $P(Y=0)=P(Y=1)=\sfrac{1}{2}$, the best regression function $s^*$ simplifies
to
\begin{equation}
s^*(\mathbf{x}) = \frac{p_{\mathbf{X}|Y}(\mathbf{x}|Y=1)}{p_{\mathbf{X}|Y}(\mathbf{x} | Y=0) + p_{\mathbf{X}|Y}(\mathbf{x} | Y=1)}.
\end{equation}
If we further assume that samples for $Y=0$ (resp. $Y=1$) are drawn from some parameterized
distribution with probability density $p_{\mathbf{X}}(\mathbf{x}|\theta_0)$ (resp. $p_{\mathbf{X}}(\mathbf{x}|\theta_1)$), then the best regression function can be rewritten
as
\begin{equation}
s^*(\mathbf{x}) = \frac{p_{\mathbf{X}}(\mathbf{x}|\theta_1)}{p_{\mathbf{X}}(\mathbf{x} | \theta_0) + p_{\mathbf{X}}(\mathbf{x} | \theta_1)}.
\end{equation}
In particular, this regression function satisfies conditions of
Thm.~\ref{thm:ratio-equivalence} since $s^*(\mathbf{x}) =
m(r(\mathbf{x}; \theta_0, \theta_1))$,
for $m(r(\mathbf{x})) = ({1 + r(\mathbf{x})})^{-1}$, is monotonic with
$r(\mathbf{x}; \theta_0, \theta_1)$.

Proposition~\ref{thm:best-classifier} holds for the squared error loss, but it
can be similarly shown that classifiers minimizing the exponential loss, the
binomial log-likelihood (or cross-entropy) or the squared hinge loss are also
monotonic with the density ratio~\citep{friedman2000additive,lin2002support}.
However, a classifier with discrete outputs and minimizing the zero-one loss
does not satisfy conditions of the theorem.

\end{document}